\documentclass[12pt]{iopart}

\expandafter\let\csname equation*\endcsname\relax
\expandafter\let\csname endequation*\endcsname\relax

\usepackage[pdftex]{graphicx}
\usepackage{here}
\usepackage{amsthm}
\usepackage{arydshln}
\usepackage{setspace}
\usepackage{cite}
\usepackage{wrapfig}
\usepackage{amsmath}
\usepackage{amssymb, latexsym, bm, mathtools, braket, multirow, url}
\usepackage{color}
\usepackage{cancel}
\usepackage{comment}
\usepackage{ascmac}
\usepackage{tikz}
\usepackage{paralist}

\usepackage[pdftex]{hyperref}

\hypersetup{
	colorlinks=true,
	linkcolor=blue,
	citecolor=red,
	bookmarks=true,
	bookmarksnumbered=true,
	pdfborder={0 0 0},
	bookmarkstype=toc
}

\usepackage[all]{xy}

\theoremstyle{plain}
\newtheorem{theorem}{Theorem}
\newtheorem{definition}[theorem]{Definition}
\newtheorem{lemma}[theorem]{Lemma}
\newtheorem{proposition}[theorem]{Proposition}
\newtheorem{fact}[theorem]{Fact}
\newtheorem{conjecture}[theorem]{Conjecture}

\theoremstyle{definition}

\newcommand{\cV}{\mathcal{V}}
\newcommand{\cK}{\mathcal{K}}
\newcommand{\cT}{\mathcal{T}}
\newcommand{\cM}{\mathcal{M}}

\newcommand{\cE}{\mathcal{E}}

\newcommand{\cB}{\mathcal{B}}
\newcommand{\cH}{\mathcal{H}}
\newcommand{\cS}{\mathcal{S}}

\newcommand{\cX}{\mathcal{X}}
\newcommand{\cP}{\mathcal{P}}
\newcommand{\cL}{\mathcal{L}}

\newcommand{\cC}{\mathcal{C}}

\newcommand{\ketbra}[2]{\ket{#1}\!\bra{#2}}

\begin{document}

\title[Pseudo standard structure cannot be distinguished from standard  structure]{Pseudo standard entanglement structure cannot be distinguished from standard entanglement structure}
\author{Hayato Arai$^1$ and Masahito Hayashi$^{2,3,4,1}$}

\address{$^1$ Graduate School of Mathematics, Nagoya University, Furo-cho, Chikusa-ku, Nagoya, 464-8602, Japan}
\address{$^2$ Shenzhen Institute for Quantum Science and Engineering, 
Southern University of Science and Technology, Nanshan District, Shenzhen, 518055, China}
\address{$^3$ International Quantum Academy (SIQA), Shenzhen 518048, China}
\address{$^4$ Guangdong Provincial Key Laboratory of Quantum Science and Engineering,
Southern University of Science and Technology, Shenzhen 518055, China}
\eads{\mailto{m18003b@math.nagoya-u.ac.jp}, \mailto{hayashi@sustech.edu.cn}}
\vspace{10pt}
\begin{indented}
\item[] June 2022
\end{indented}

\begin{abstract}
An experimental verification of the maximally entangled state ensures that the constructed state is close to the maximally entangled state,
but it does not guarantee that the state is exactly the same as the maximally entangled state.
Further, the entanglement structure is not uniquely determined in general probabilistic theories even if we impose that the local subsystems are fully equal to quantum systems.
Therefore, the existence of the maximally entangled state depends on 
whether the standard entanglement structure is valid.
To examine this issue, we introduce pseudo standard entanglement structure as a structure of quantum composite system under natural assumptions based on
the existence of projective measurements
and the existence of approximations of all maximally entangled standard states.
Surprisingly, there exist infinitely many pseudo standard entanglement structures different from the standard entanglement structure.
In our setting, any maximally entangled state can be arbitrarily approximated by an entangled state that belongs to our obtained pseudo standard entanglement structure.
That is, experimental verification does not exclude the possibility of our obtained pseudo standard entanglement structure that is different from the standard entanglement structure.
On the other hand,
such pseudo structures never possess global unitary symmetry,
i.e., global unitary symmetry is essential condition for the standard entanglement structure.
\end{abstract}

\vspace{2pc}


\maketitle

\section{Introduction}

Recently, many studies discussed verification of maximally entangled states 
from theory \cite{HMT,Ha09-2,HayaM15,PLM18,ZH4,Markham} to experiment \cite{MST,KSKWW,Bavaresco,FVMH,JWQ}.
However, their verification ensures only that the constructed state is close to 
the maximally entangled state.
Therefore, it does not guarantee that 
it is exactly the same as the maximally entangled state.
That is, such an experimental verification does not necessarily support the existence of the maximally entangled state.
Hence, it is impossible to experimentally verify
the \textit{standard entanglement structure} (SES)
of the composite system,
in which a state on the composite system
is given as a normalized positive semi-definite matrix on the tensor product space
even if the local systems are fully equal to standard quantum theory.

Furthermore,
a theoretical structure of quantum bipartite composite systems
is not uniquely determined
even if
we impose that the local subsystems are exactly the same as standard quantum subsystems
\cite{Janotta2014,Lami2017,Aubrun2020,Plavala2021,Arai2019,YAH2020,HK2022,ALP2019,ALP2021}.
This problem is recently studied in the modern operational approach of foundations of quantum theory, called \textit{General Probabilistic Theories} (GPTs) \cite{Janotta2014,Lami2017,Aubrun2020,Plavala2021,Arai2019,YAH2020,HK2022,ALP2019,ALP2021,
Kimura2010, Bae2016, PR1994,Pawlowski2009,Short2010,Barnum2012,Plavala2017,Matsumoto2018,Takagi2019,Yoshida2020,CDP2010,Spekkens2007,MAB2022,CS2015,CS2016,CS2015-2,
Muller2013,BLSS2017,Barnum2019,Janotta2013}.
GPTs start with fundamental probabilistic postulates to define states and measurements.
Even though the postulates and the mathematical definition of GPTs are physically reasonable,
GPTs cannot uniquely determine the model of the bipartite quantum composite system
even if the subsystems are equivalent to standard quantum systems.
For example, GPTs allow the model with no entangled states
as well as the model with ``strongly entangled'' states than the standard quantum system
in addition to the SES \cite{Janotta2014,Lami2017,Aubrun2020,Plavala2021,Arai2019,YAH2020,HK2022,ALP2019,ALP2021}.

Some studies of GPTs deal with the most general models satisfying fundamental probabilistic postulates \cite{PR1994,Pawlowski2009,Short2010,Barnum2012,Plavala2017,Matsumoto2018,Takagi2019,Yoshida2020,CDP2010,Spekkens2007,MAB2022,CS2015,CS2016,CS2015-2},
and they investigate physical or informational properties in general.
Because general models are sometimes quite different from standard quantum systems,
general properties behave unlike present experimental facts \cite{PR1994,Pawlowski2009,Short2010,Barnum2012,Plavala2017,Spekkens2007}.
While the above studies aim to investigate physical and informational properties in general models,
our interest is whether there exists a quantum-like model satisfying present experimental facts except for standard quantum systems.
That is, this paper aims to impose several conditions on GPTs for behaving in a similar way as quantum theory
and to exprole the existence of other GPTs to satisfy these conditions.

As the first condition,
this paper deals with a class of the general models
called
\textit{bipartite entanglement structures with local quantum systems} \cite{Janotta2014,Lami2017,Aubrun2020,Plavala2021,Arai2019,YAH2020,HK2022,ALP2019,ALP2021}(hereinafter, we simply call it \textit{Entanglement Structures} or ESs)
\footnote{
\textit{Entanglement} is a concept defined not only in quantum composite systems but also in general models
whose local subsystems are not necessarily equal to standard quantum systems \cite{ALP2019,ALP2021,CS2016,CS2015-2}.
However, our interest is  a ``similar structure'' to standard quantum entanglement;
therefore, we impose that the local subsystems are equal to standard quantum systems, as we mentioned.
},
i.e.,
we deal with the composite models in GPTs with the assumption that their local systems are completely equivalent to standard quantum systems.
As many studies pointed out,
many models satisfy this condition \cite{Janotta2014,Lami2017,Aubrun2020,Plavala2021,Arai2019,YAH2020,HK2022,ALP2019,ALP2021},
and some models do not behave in a similar way as quantum theory \cite{Lami2017,Aubrun2020,Arai2019,YAH2020}.
Therefore, as present experimental facts,
this paper mainly focuses on two additional conditions, \textit{undistinguishability} and \textit{self-duality}, mentioned below.

The second condition, undistinguishability,
is introduced as the possibility of the verification of maximally entangled states with tiny errors.
The error probability of verification is upper bounded by using the trace norm
due to a simple inequality.
Therefore, we mathematically define $\epsilon$-undistinguishability as $\epsilon$-upper bound of a distance based on trace norm between the state space in an ES and the set of maximally entangled states.
If an ES satisfies undistinguishability with enough small errors,
it cannot be denied by physical experiments of verification of maximally entangled states
that our physical system might obey the structure (not standard one).

The third condition, self-duality,
is defined as the equality between the state space and the effect space in an ES.
This paper introduces self-duality as a saturated situation of \textit{pre-duality},
and we point out the correspondence between pre-duality and \textit{projectiviy}.
Projectivity is one of the postulates in standard quantum theory \cite{Neumann1932,Luders1951,Davies1970,Ozawa1984},
which ensures a measurement whose post-measurement states are given as the normalization of its effects.
Self-duality is a saturation of projectivity and a common property that classical and quantum theory possess.
Moreover, self-duality is important for deriving algebraic structures in physical systems.
When a self-dual model satisfies a kind
of strong symmetry, called homogeneity, 
the state space is characterized by Jordan Algebras \cite{Jordan1934,Koecher1957,Barnum2019,BMA2020},
which leads to essentially limited types of models, including classical and quantum theory \cite{Jordan1934,Koecher1957}.

In summary, this paper aims to discuss whether there exists an ES with $\epsilon$-undistinguishability and self-duality other than the SES.
Such a structure cannot be distinguished by any verification of maximally entangled states with errors larger than $\epsilon$ and satisfies saturated projectivity.
Due to this physical similarity, we call an ES with $\epsilon$-undistinguishability and self-duality an \textit{$\epsilon$-Pseudo Standard Entanglement Structure ($\epsilon$-PSES)},
and our main question is whether there exists an $\epsilon$-PSES other than the SES, especially for small $\epsilon$.
Surprisingly,
we show that there exists infinitely many $\epsilon$-PSESs for any $\epsilon>0$.
In other words,
there exist infinite possibilities of ESs that cannot be distinguished from the SES by physical experiments of verification of maximally entangled states
even though the error of verification is extremely tiny and even though we impose projectivity.

In the next step, we explore the operational difference between PSESs and the SES in contrast to the physical similarity between $\epsilon$-PSESs and the SES.
For this aim, this paper focuses on the performance of perfect state discrimination,
and we show the infinite existence of $\epsilon$-PSESs
that have two perfectly distinguishable non-orthogonal states.
Perfect distinguishability in GPTs has been studied well
\cite{Arai2019,YAH2020,Kimura2010,Bae2016,Muller2013,BLSS2017,Barnum2019}.
For example,
while perfect distinguishablity is equivalent to orthogonality in quantum and classical theory, 
the reference \cite{Muller2013} has implied that orthogonality is a sufficient condition for perfectly distinguishablity in any self-dual model under a specific condition.
Also, the reference \cite{Arai2019,YAH2020} has shown that a non-self-dual model of quantum composite systems has a distinguishable non-orthogonal pair of two states.
Therefore, it is interesting to consider
whether non-self-duality is necessary for non-orthogonal perfect distinguishability.
In this paper, we negatively solve this problem, i.e.,
we show that infinitely many $\epsilon$-PSESs with non-orthogonal distinguishability.
In other words,
some $\epsilon$-PSESs have superiority over the SES in perfect discrimination
even though $\epsilon$-PSESs cannot be distinguished from the SES by verification tasks with errors.

Finally, since the SES cannot be distinguished from ESs based on present experimental facts, $\epsilon$-undistinguishability and self-duality, 
we focus on another condition, symmetry conditions. That is, we investigate what symmetry condition determines the SES.
Symmetric conditions cannot be observed directly,
but a symmetric condition plays an important role in characterizing models corresponding to Jordan Algebras out of general models of GPTs \cite{Muller2013,BLSS2017,Barnum2019}.
In this paper,
restricting the characterization to the class of ESs,
we determine the SES out of ESs by a condition about the global unitary group, which is smaller than the group in \cite{Muller2013,BLSS2017,Barnum2019}.
As a result,
we clarify that global unitary symmetry is an essential property of the SES.

The remaining part of this paper is organized as follows.
First, we introduce the mathematical definition of models and composite systems in GPTs,
and
we see non-uniqueness of models of the quantum composite systems,
i.e.,
any model satisfying the inclusion relation \eqref{eq:quantum} is regarded as the quantum composite systems in section~\ref{sect.definition}.
Next, we introduce ESs and the standard entanglement structure in section~\ref{sect.ses}.
In this section, we discuss the condition when an ES cannot be distinguished from the SES,
and we introduce $\epsilon$-undistinguishable condition.
Next, we introduce pre-duality and self-duality
as consequences of projectivity in section~\ref{sect.self-dual}.
Also, we introduce a PSES as an entanglement structure with self-duality and $\epsilon$-undistinguishable condition.
Section~\ref{sect.hierarchy} establishes a general theory for the construction of self-dual
models.
We show that any pre-dual model can be modified to a saturating model self-duality (theorem~\ref{theorem:sd}, theorem~\ref{theorem:hie1}).
In section~\ref{sect.construct}, we apply the above general theory to the quantum composite system.
We show the existence of infinitely many examples of PSESs (theorem~\ref{theorem:main}).
Also, we show that the PSESs have non-orthogonal perfectly distinguishable states (theorem~\ref{theorem:dist}) in section~\ref{sect.discrimination}.
Further, we discuss the characterization of the SES with group symmetric conditions in section~\ref{sect.symmetry}.
Finally, we
summarize our results
and give an open problem in section~\ref{sect.conclude}.
In this paper, detailed proofs of some results are written in appendix.

\begin{table*}[htb]
	\caption{Notations}
	\centering
	\begin{tabular}{clc}
	\hline
	notation & meaning & equation \\ \hline \hline
	$\cS(\cK,u)$ & the state space of the model $\cK$ with the unit $u$ &\eqref{def:state}\\
	$\cE(\cK,u)$ & the effect space of the model $\cK$ with the unit $u$  &\eqref{def:eff} \\
	$\cM(\cK,u)$ & the measurement space of the model $\cK$ with the unit $u$  &\eqref{def:mea} \\
	$\cT(\cH)$ & the set of all Hermitian matrices on a Hilbert space $\cH$ &-\\
	\multirow{2}{*}{$\cT_+(\cH)$} & the set of all Positive semi-definite matrices&\multirow{2}{*}{-}\\
	&\multicolumn{1}{r}{ on a Hilbert space $\cH$ }&\\
	$\cK_1\otimes\cK_2$ \quad& the tensor product of positive cones & \eqref{eq:tensor} \\ 
	$\mathrm{SEP}(A;B)$\quad & the positive cone that has only separable states &\eqref{eq:sep}\\
	$\mathrm{SES}(A;B)$\quad & the standard entanglement structure &\eqref{eq:SES}\\
	$\mathrm{ME}(A;B)$\quad & the set of all maximally entangled states &-\\
	\multirow{2}{*}{$D(\cK\|\sigma)$\quad }& the distance between an entanglement structure $\cK$ &\multirow{2}{*}{\eqref{def:distance1}}\\
	&\multicolumn{1}{r}{  and a state $\sigma$}&\\
	$D(\cK_1\|\cK_2)$\quad & the distance between entanglement structures $\cK_1$ and $\cK_2$ &\eqref{def:distance2}\\
	\multirow{2}{*}{$D(\cK)$\quad} & the distance between an entanglement structure $\cK$ &\multirow{2}{*}{\eqref{def:distance}}\\
	&\multicolumn{1}{r}{ and the SES} &\\
	$\tilde{\cK}$\quad & a self-dual modification of pre-dual cone $\cK$ \quad &-\\
	$\mathrm{MEOP}(A;B)$\quad & the set of maximally entangled orthogonal projections&\eqref{eq:proj}\\
	$\mathrm{NPM}_r(A;B)$\quad & a set of non-positive matrices &\eqref{def:NPM}\\
	$\cK_r(A;B)$\quad & a set of non-positive matrices with parameter $r$ &\eqref{def:Kr}\\
	$r_0(A;B)$\quad & the parameter given in proposition~\ref{prop:construction1} &\eqref{def:r0}\\
	\multirow{2}{*}{$\cP_0(\vec{P})$\quad} & a family belonging to $\mathrm{MEOP}(A;B)$ &\multirow{2}{*}{\eqref{def:PE}}\\
	&\multicolumn{1}{r}{ defined by a vector $\vec{P}\in\mathrm{MEOP}(A;B)$ }&\\
	$\mathrm{GU}(A;B)$\quad & the group of global unitary maps&\eqref{eq:gu}\\
	$\mathrm{LU}(A;B)$\quad & the group of local unitary maps&\eqref{eq:lu}\\
	\multirow{2}{*}{$N(r;\{E_k\})$\quad} & a non-positive matrix with a parameter $r\ge0$&\multirow{2}{*}{\eqref{def:Nr}}\\
	&\multicolumn{1}{r}{  and a family $\{E_k\}\in\mathrm{MEOP}(A;B)$ }&\\
	\hline
	\end{tabular}
\end{table*}

\section{GPTs and Composite systems}\label{sect.definition}

At the beginning, we simply introduce the concept of GPTs, which is a generalization of classical and quantum theory.
We consider a finite-dimensional general model that contains states and measurements.
Because any randomization of two states is also a state,
state space must be convex.
A measurement is an operation over a state to get an outcome $\omega$ with a certain probability
dependent on the given state and the way of the measurement.
Mathematically, this concept defines a measurement as
a family of functional from state space to $[0,1]$, whose output corresponds to the probability.
Also, any randomization of two measurements is also a measurement;
therefore, measurement space must be convex.

As a consequence of the above assumptions,
a model of GPTs is defined by the following mathematical setting.
Let  $\cV$ be a real vector space with an inner product $\langle,\rangle$.
We call $\cK\subset\cV$ a positive cone if $\cK$ satisfies the following three conditions:
$\cK$ is a closed convex set, $\cK$ has an inner point, and $\cK\cap(-\cK)=\{0\}$.
Also, we define the dual cone $\cK^\ast$ for a positive cone $\cK$ as $\cK^\ast:=\{x\in\cV\mid \langle x,y\rangle\ge0\ \forall y\in\cK \}$.
Then, a model of GPTs is defined as a tuple $(\cV,\cK,u)$,
where $u$ is a fixed inner point in $\cK^\ast$.
In a model of GPTs, the state space, the effect space, and the measurement space are defined as follows.
The state space $\cS(\cK,u)$ of $(\cV,\cK,u)$ is defined as
\begin{align}\label{def:state}
	\cS(\cK,u):=\{\rho\in\cK\mid\langle \rho,u\rangle=1\},
\end{align}
and an extremal point of $\cS(\cK,u)$ is called a pure state.
Also, the effect space $\cE(\cK,u)$ and the measurement space $\cM(\cK,u)$ of $(\cV,\cK,u)$ is respectively defined as
\begin{align}
	\cE(\cK,u):&=\left\{e\in\cK^\ast\mid 0\le\langle e,\ \rho\rangle\le1 \ \forall \rho\in\cS(\cK,u)\right\},\label{def:eff}\\
	\cM(\cK,u):&=\left\{\{M_{\omega}\}_{\omega\in\Omega}\middle| M_\omega\in\cK^\ast, \ \sum_{\omega\in\Omega} M_\omega=u\right\},\label{def:mea}
\end{align}
where $\Omega$ is the finite set of outcome.
Also, an extremal element $e\in\cE(\cK,u)$ is called a pure effect.
Besides, the probability to get an outcome $\omega$ is given by $\langle \rho,M_\omega\rangle$ for a state $\rho\in\cS(\cK,u)$ and a measurement $\{M_\omega\}_{\omega\in\Omega}\in\cM(\cK,u)$.
Here, we remark the definition of a measurement.
A vector space and its dual space are mathematically equivalent when the dimension of vector space is finite.
In this paper, the effect space and the measurement space are defined as subsets of the original vector space for later convenience.

The above mathematical setting is a generalization of classical and quantum theory.
For example, the model of quantum theory is given by the model $(\cT(\cH),\cT_+(\cH),I)$,
where $\cT(\cH)$, $\cT_+(\cH)$, and $I$ are denoted as the set of all Hermitian matrices on Hilbert space $\cH$, the set of all positive semi-definite (PSD) matrices on $\cH$, and the identity matrix on $\cH$, respectively.
Then, the state space $\cS(\cT_+(\cH),I)$ and the measurement space $\cM(\cT_+(\cH),I)$ are equal to the set of all density matrices and the set of all positive operator valued measures (POVMs), respectively.
This is because the dual $\cT_+(\cH)^\ast$ is equal to itself.
This property $\cK^\ast=\cK$ is called self-duality, as we mention later.
In this way, the model $(\cT(\cH),\cT_+(\cH),I)$ is regarded as the model of quantum theory.

Next, we define a model of composite systems in GPTs.
We say that a model $(\cV,\cK,u)$ is a model of the composite system of two submodels $(\cV_A,\cK_A,u_A)$ and $(\cV_B,\cK_B,u_B)$
when the model $(\cV,\cK,u)$ satisfies the following three conditions:
(i) $\cV=\cV_A\otimes\cV_B$,
(ii) $\cK_A\otimes\cK_B\subset\cK\subset(\cK_A^\ast\otimes\cK_B^\ast)^\ast$,
and (iii) $u=u_A\otimes u_B$.
Here, the tensor product of two cones $\cK_A\otimes\cK_B$ is defined as
\begin{align}\label{eq:tensor}
	\cK_A\otimes\cK_B:=\left\{\sum_k a_k\otimes b_k\middle| a_k\in\cK_A,\ b_k\in\cK_B\right\}.
\end{align}
This definition derives from the following physical reasonable assumption.
The composite system contains Alice's system $(\cV_A,\cK_A,u_A)$ and Bob's system $(\cV_B,\cK_B,u_B)$.
It is natural to assume that Alice and Bob can prepare local states $\rho_A\in\cS(\cK_A,u_A)$ and $\rho_B\in\cS(\cK_B,u_B)$ independently.
Consequently, the product state $\rho_A\otimes\rho_B$ is prepared in the composite system (figure~\ref{figure-composite}),
i.e.,
the global state space $\cS(\cK,u_A\otimes u_B)$ contains the product state $\rho_A\otimes\rho_B$.
This scenario implies the inclusion $\cK_A\otimes\cK_B\subset\cK$.
Similarly, the product effect $e_A\otimes e_B$ can also be prepared in the composite system.
This scenario also implies $\cK_A^\ast\otimes\cK_B^\ast\subset\cK^\ast$,
which is rewritten as $\cK\subset(\cK_A^\ast\otimes\cK_B^\ast)^\ast$.
We give another scenario that derives the definition of models of composite systems in appendix~\ref{append-com}.

\begin{figure}[t]
	\centering
	\includegraphics[width=8cm]{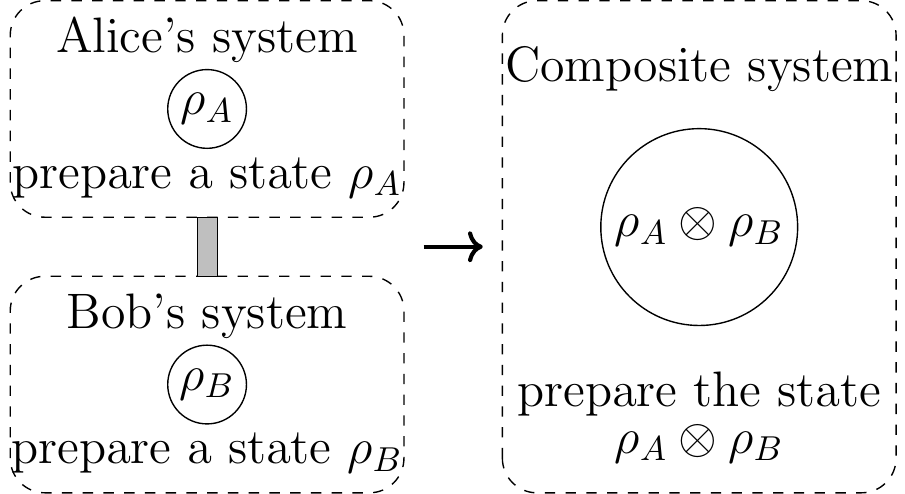}
	\caption{
	When local states $\rho_A$ and $\rho_B$ are prepared individually by Alice and Bob,
	the product states $\rho_A\otimes\rho_B$ is prepared on composite system.
	}
	\label{figure-composite}
\end{figure}

\section{Entanglement structures and the standard entanglement structure}\label{sect.ses}

Now, let us consider the composite system of two quantum subsystems $(\cT(\cH_A),\cT_+(\cH_A),I_A)$ and $(\cT(\cH_B),\cT_+(\cH_B),I_B)$.
An entanglement structure,
i.e.,
a model of the composite system is given as $(\cT(\cH_A\otimes \cH_B),\cK,I_{A;B})$ that satisfies
\begin{gather}
	\mathrm{SEP}(A;B)\subset\cK\subset\mathrm{SEP}^\ast(A;B),\label{eq:quantum}\\
	\mathrm{SEP}(A;B):=\cT_+(\cH_A)\otimes\cT_+(\cH_B)\label{eq:sep}.
\end{gather}
The cone $\mathrm{SEP}(A;B)$ corresponds to the model that has only separable states,
but the model has beyond-quantum measurements that can discriminate non-orthogonal separable states \cite{Arai2019}.
Also, the cone $\mathrm{SEP}^\ast(A;B)$ corresponds to the model that has elements in $\mathrm{SEP}^\ast(A;B)\setminus\cT_+(\cH_A\otimes\cH_B)$, which are regarded as more strongly entangled elements.
It is believed that actual quantum composite systems obey the model $\cT_+(\cH_A\otimes\cH_B)$,
and an important aim of studies of GPTs is to characterize this model.
Hereinafter, we call this model \textit{standard entanglement structure} (SES),
and we use the notation 
\begin{align}\label{eq:SES}
	\mathrm{SES}(A;B):=\cT_+(\cH_A\otimes\cH_B).
\end{align}
In this way, a model of composite systems is not uniquely determined in general,
i.e.,
there are many possible entanglement structures of the composite system in GPTs.

Next, to consider the experimental verification of a given model,
we introduce the distinguishability of two state spaces of two given models $\cK_1$ and $\cK_2$.
Because any Hermitian matrices $X$, $\rho$, and $\sigma$ satisfy the inequality
\begin{align}
	\left|\Tr X\rho-\Tr X\sigma\right|\le\|X\|_\infty\|\rho-\sigma\|_1,
\end{align}
this paper estimates the error probability of verification tasks by trace norm,
where $\|\ \|_\infty$ is spectral norm.
Therefore,
given a state $\sigma\in\cS(\cK_2,u_2)$, the quantity 
\begin{align}\label{def:distance1}
D(\cK_1\|\sigma):= \min_{\rho \in \cS(\cK_1,u_1)}\| \rho-\sigma\|_1
\end{align}
expresses how well 
the state $\sigma$ is distinguished from states in $\cK_1$.
Optimizing the state $\sigma$, we consider the quantity
\begin{align}\label{def:distance2}
D(\cK_1\|\cK_2):=\max_{\sigma \in \cS(\cK_2,u_2)} D(\cK_1\|\sigma),
\end{align}
which expresses the optimum distinguishability of the model $\cK_2 $
from the model $\cK_1$.
Hence, the quantity $D(\mathrm{SES}(A;B) \|\cK)$ expresses
how the standard model $\mathrm{SES}(A;B)$
can be distinguished from a model $\cK$.

However, we often consider the verification of 
a maximally entangled state
because a maximally entangled state is the furthest state from separable states.
In order to consider maximally entangled states,
we assume that $\dim(\cH_A)=^dim(\cH_B)=d$ in the following discussion.
When the range of the above maximization \eqref{def:distance2} is restricted to maximally entangled states, 
the distinguishability of the standard model $\mathrm{SES}(A;B)$ from
the model $\cK$ is measured by the following quantity:
\begin{align}\label{def:distance}
D(\cK)
:= &\max_{\sigma \in \mathrm{ME}(A;B) } D(\cK\|\sigma),
\end{align}
where the set $\mathrm{ME}(A;B)$ is denoted as the set of all maximally entangled states on $\cH_A\otimes\cH_B$.
Given a model $\cK$,
we introduce \textit{$\epsilon$-undistinguishable condition} as
\begin{align}\label{cd:epsilon}
	D(\cK)\le\epsilon.
\end{align}
That is, if a model $\cK$ satisfies $\epsilon$-undistinguishablity,
even when we pass the verification test for any maximally entangled state,
we cannot deny the possibility that our system is the model $\cK$ (figure~\ref{figure-cverification}).
Clearly, there are many models satisfying this condition.
For example,
$\mathrm{SEP}^\ast$ satisfies it because $D(\mathrm{SEP}^\ast)=0$.
In other words, it is impossible to deny such a possibility without assuming an additional constraint for our model.
The aim of this paper is to examine whether there exists a natural condition to deny $\epsilon$-undistinguishablity.
As a natural condition, the next section introduces self-duality
via projective measurements.

\begin{figure}[t]
	\centering
	\includegraphics[width=12cm]{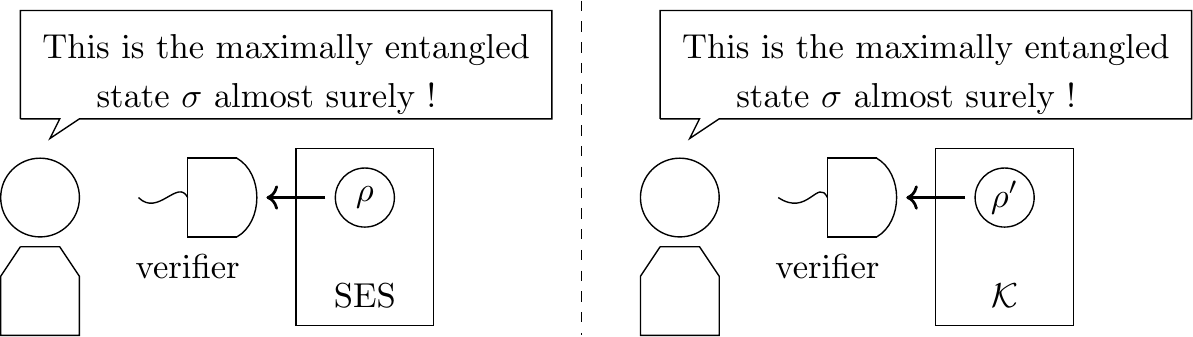}
	\caption{
	Even if the verifier's system is subject to an $\epsilon$-undistinguishable entanglement structure $\cK\neq\mathrm{SES}(A;B)$,
	the verifier achieves the verification task of a given maximally entangled state $\sigma$ with error $\epsilon$ by preparing a state $\rho'\in\cK$ satisfying $\|\rho'-\sigma\|_1\le\epsilon$.
	In this sense, such verification tasks can not distinguish the entanglement structures $\mathrm{SES}(A;B)$ and $\cK$ when $\cK$ satisfies $\epsilon$-undistinguishability.
	}
	\label{figure-cverification}
\end{figure}

\section{Projective measurement, self-duality, and pseudo standard entanglement structures}\label{sect.self-dual}
Next, we introduce pre-duality and self-duality via projective measurements.
In standard quantum theory,
there exists a measurement $\{e_i\}_{i\in I}$ such that the post-measurement state with the outcome $i$ is given as $e_i/\Tr e_i$ independently of the initial state when the effect $e_i$ is pure.
Such a measurement is called a projective measurement \cite{Neumann1932,Davies1970,Ozawa1984}.
The measurement projectivity is one of the postulates of standard quantum theory \cite{Neumann1932,Davies1970,Ozawa1984}.
Therefore, in this paper, we impose that any model $\cK$ satisfies the following condition:
for any pure effect $e\in\cE(\cK,u)$, there exists a measurement $\{e_i\}$ such that
an element $e_{i_0}$ is equal to $e$, and the post-measurement state is given as $\overline{e_{i_0}}:=e_{i_0}/\Tr e_{i_0}$.

Because any effect satisfies the condition $0\le\langle e,\rho\rangle\le1\ \forall \rho\in\cS(\cK,u)$ in \eqref{def:eff},
the element $u-e$ also belongs to $\cE(\cK,u)$,
which implies that
the family $\{e,u-e\}$ belongs to $\cM(\cK,u)$ for any effect $e\in\cE(\cK,u)$.
Also,
pure effects span the effect space $\cE(\cK,u)$ with convex combination,
and the effect space $\cE(\cK,u)$ generates the dual cone $\cK^\ast$ with constant time.
Therefore, the existence of projective measurement implies the inclusion relation $\cK\supset\cK^\ast$.
In this paper, this property $\cK\supset\cK^\ast$ is called pre-duality.
\begin{figure}[t]
	\centering
	\includegraphics[width=4cm]{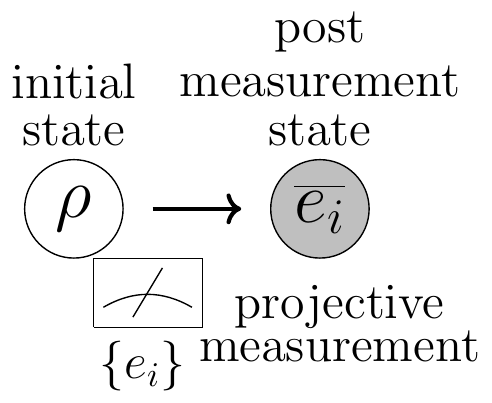}
	\caption{
	When a projective measurement $\{e_i\}$ is applied to the system with an initial state $\rho$,
	we obtain an outcome $i$ and the corresponding post-measurement state $\overline{e_i}=e_i/\Tr e_i$ independent of the initial state $\rho$.
	}
	\label{figure-projection}
\end{figure}

Here, we remark on the relation between projectivity and repeatability.
Repeatability is a postulate of standard quantum theory, sometimes included in the projection postulate \cite{Neumann1932,Chefles2003,Buscemi2004,CY2014,CY2016,Perinotti2012,SSKH2021}.
Repeatability ensures that 
the same effect is observed with probability 1 in the sequence of the same measurements,
and the effects do not change the post-measurement state (figure~\ref{figure-repeatable}).
\begin{figure}[t]
	\centering
	\includegraphics[width=6cm]{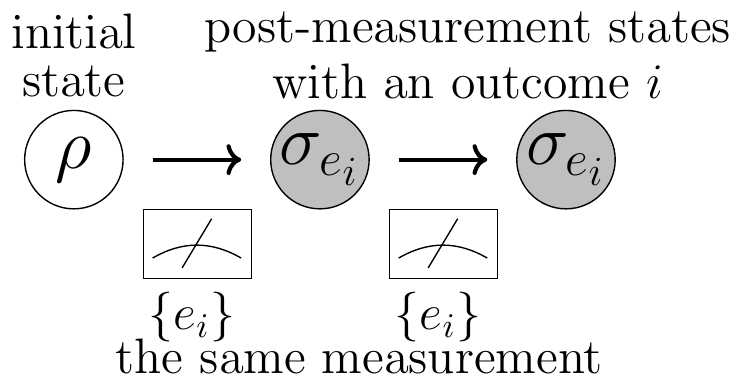}
	\caption{When an initial state is measured by a measurement $\{e_i\}$ twice,
	the post-measurement states of first and second measurement with an outcome $i$ are equivalent.
	}
	\label{figure-repeatable}
\end{figure}
In standard quantum theory,
repeatability is sometimes confounded with the above condition that any pure effects can constructs a projective measurement.
However, similarly to the reference \cite{Neumann1932,Chefles2003,Buscemi2004,CY2014,CY2016,Perinotti2012,SSKH2021} and the translated version \cite[Discussion]{Luders1951}, the above concept of repeatability implies the following more weaker condition than the existence of projective measurements.
Repeatability requires that the tuple of post-measurement states $\{\sigma_{e_i}\}_{i\in I}$ is perfectly distinguishable by the measurement $\{e_i\}_{\i\in I}$,
i.e.,
the equation $\Tr \sigma_{e_i} e_j=\delta_{i,j}$.
In other words,
repeatability requests the $|I|$ number of constraints for the post-measurement state $\sigma_{e_i}$.
On the other hand,
projectivity determines post-measurement states completely.
In other words,
projectivity requests the same number of constraints for the post-measurement state as the dimension of $\cK^\ast$.
In general, the number of outcomes $|I|$ is smaller than the dimension of $\cK^\ast$;
therefore, projectivity is a stronger postulate than repeatability in terms of the number of constraints.

Now, we consider pre-dual models of composite systems $(\cT(\cH_A\otimes \cH_B),\cK,I_{A;B})$.
For example,
let us consider the model that contains only separable measurements.
In such a model, the dual cone is given as $\cK^\ast=\mathrm{SEP}(A;B)$,
and therefore, the model satisfies $\cK=\mathrm{SEP}(A;B)^\ast$ because the dual of a dual cone is equal to the original cone.
However, the state space $\cS(\mathrm{SEP}(A;B)^\ast,I_{A;B})$ has excessive many states;
the state space $\cS(\mathrm{SEP}(A;B)^\ast,I_{A;B})$ has not only all quantum states but also all entanglement witnesses with trace 1.
Then, there exist two state $\rho_1,\rho_2\in\cS(\cK,I_{A;B})$ such that they satisfy $\Tr\rho_1\rho_2<0$.
Not only the case $\cK=\mathrm{SEP}(A;B)^\ast$,
but also any pre-dual model has two states $\rho_1,\rho_2$ with $\Tr\rho_1\rho_2<0$ unless $\cK=\cK^\ast$.
In this way, pre-dual models have a gap between the state space and the effect space unless $\cK=\cK^\ast$.
In order to remove such a gap,
as the saturated situation of pre-duality,
we extend the measurement effect space and restrict the state space by modifying the cone to $\tilde{\cK}$ with satisfying $\tilde{\cK}\supset\tilde{\cK}^\ast$.
Here, we denote the modified model as $(\cT(\cH_A\otimes \cH_B),\tilde{\cK},I_{A;B})$,
and we say that the model is \textit{self-dual} if the cone $\tilde{\cK}$ satisfies $\tilde{\cK}^\ast=\tilde{\cK}$.

Self-duality denies entanglement structures whose effect space is strictly larger than state space,
for example, $\mathrm{SEP}^\ast(A;B)$.
In this paper,
in order to investigate quantum-like entanglement structures,
we consider the combination of $\epsilon$-undistinguishable condition and self-duality.
Hereinafter,
we say that an entanglement structure $\cK$ is an \textit{$\epsilon$-pseudo standard entanglement structure} ($\epsilon$-PSES)
if $\cK$ satisfies $\epsilon$-undistinguishable condition and self-duality.
A typical example of $\epsilon$-PSESs is, of course, the SES,
but
another example of $\epsilon$-PSESs is not known,
especially in the case when $\epsilon$ is very small.
If there exists another $\epsilon$-PSES for small epsilon,
the model cannot be distinguished from the SES by physical experiments of verification of maximally entangled states with small errors even though we impose projectivity.
In this paper,
we investigate the problem of whether there exists another example of $\epsilon$-PSESs,
especially another entanglement structure $\cK$ with self-duality and \eqref{cd:epsilon}.
As a result,
we give an infinite number of examples of PSESs by applying a general theory given in the next section.

\section{Self-dual modification and hierarchy of pre-dual cones with symmetry under operations}\label{sect.hierarchy}

In this section,
we state general theories to show the existence of self-dual models satisfying \eqref{eq:quantum}.
The first result is that any pre-dual model can always be modified to a self-dual model.
\begin{theorem}[self-dual modification]\label{theorem:sd}
	Let $\cK$ be a pre-dual cone in $\cV$.
	Then, there exists a positive cone $\tilde{\cK}$ such that 
	\begin{align}\label{eq:sd}
		\cK\supset\tilde{\cK}=\tilde{\cK}^\ast\supset\cK^\ast.
	\end{align}
\end{theorem}
A self-dual cone $\tilde{\cK}$ to satisfy \eqref{eq:sd} is called a self-dual modification (SDM) of $\cK$.
Here, we remark that the reference \cite{BF1976} has also shown a result essentially similar to Theorem~\ref{theorem:sd}.
In the reference \cite{BF1976}, a cone is defined as a closed convex set satisfying only the property that $rx\in\cC$ for any $r\ge0$ and any $x\in\cC$.
This paper assumes additional properties, $\cC$ has non-empty interior and $\cC\cap(-\cC)=\{0\}$.
Actually, we can easily modify the proof in \cite{BF1976} in our definition,
but this thesis gives another proof
for reader's convenience in appendix~\ref{append-2}.

Also, we remark on the difference between self-dual modification and self-dualization in \cite{Janotta2013}. 
The reference \cite{Janotta2013} has shown that the state space and the effect space
of any model can be transformed by a linear homomorphism from one to another, where the effect space is considered as the subset of $\cV^\ast$. 
The result \cite{Janotta2013} can be interpreted in our setting as follows; The state space and effect space become equivalent by changing the inner product.
This process is called self-dualization in \cite{Janotta2013}. 
However, our motivation is constructing models of
the composite system with keeping the inner product to be the product form of  
the inner products in the models of the subsystems. 
Therefore, the result \cite{Janotta2013} cannot be used for our purpose.


A given pre-dual cone does not uniquely determine SDM
because the proof of Theorem~\ref{theorem:sd} and the proof in \cite{BF1976} are neither constructive nor deterministic.
Indeed, even when two self-dual cones are self-dual modifications of different pre-dual cones, they are not necessarily different self-dual cones in general.
For example, when we have three different self-dual cones $\cK_1,\cK_2,\cK_3$,
then $\cK_1+\cK_2$ and $\cK_2+\cK_3$ are pre-dual cones,
but $\cK_2$ is regarded as a modification of $\cK_1+\cK_2$ and $\cK_2+\cK_3$.
Hence, the following two concepts are useful to clarify the difference among self-dual modifications.
\begin{definition}[$n$-independence]
	For a natural number $n$, we say that a family of sets $\{\cK_i\}_{i=1}^n$ is $n$-independent if no sets $\cK_i\ (1\le i\le n)$ satisfy that $\cK_i \subset \sum_{j\neq i} \cK_j$.
	Especially, we say that $\{\cK_i\}_{i=1}^n$ is $n$-independent family of cones when any $\cK_i$ is a positive cone.
\end{definition}
\begin{definition}[exact hierarcy with depth $n$]
	For a natural number $n$, we say that pre-dual cone $\cK$ has an exact hierarchy with depth $n$
	if there exists a family of sets $\{\cK_i\}_{i=1}^n$ such that
	\begin{align}
			\cK\supset\cK_1\supsetneq\cK_2\supsetneq\cdots\supsetneq\cK_n
			\supset\cK_n^\ast\supsetneq\cdots\supsetneq\cK_1^\ast\supset\cK^\ast.
	\end{align}
	Especially, we say that $\{\cK_i\}_{i=1}^n$ is exact hierarchy of cones when any $\cK_i$ is a positive cone.
\end{definition}

Then, as an extension of theorem~\ref{theorem:sd},
the following theorem shows the equivalence between
the existence of an $n$-independent family of self-dual cones
and the existence of an exact hierarchy of pre-dual cones with depth $n$.
\begin{theorem}\label{theorem:hie1}
	Let $\cK$ be a positive cone.
	The following two statements are equivalent:
	\begin{enumerate}
		\item there exists an exact hierarchy of pre-dual cones $\{\cK_i\}_{i=1}^n$ satisfying $\cK\supset\cK_i\supset\cK^\ast$.
		\item there exists an $n$-independent family of self-dual cones $\{\cL_i\}_{i=1}^n$ satisfying that $\cL_i$ is a self-dual modification of $\cK_i$,
		i.e., 
		$\cL_i$ is a self-dual cone satisfying $\cK_i\supset\cL_i\supset\cK_i^\ast$.
	\end{enumerate}
\end{theorem}
The proof of theorem~\ref{theorem:hie1} is written in appendix~\ref{append-hie}.

\section{Existence of infinite $\epsilon$-PSESs}\label{sect.construct}
In this section,
in order to discuss the existence of $\epsilon$-PSESs,
we apply theorem~\ref{theorem:hie1} to ESs on $\cH_A\otimes\cH_B$ with $\dim(\cH_A)=\dim(\cH_B)=d$ (due to the $\epsilon$-undistinguishable condition).
As a result, we show that there exist infinitely many exactly different $\epsilon$-PSESs.


First,
we denote
$\mathrm{MEOP}(A;B)$
as the set of all maximally entangled orthogonal projections on $\cH_A\otimes\cH_B$,
i.e.,
\begin{align}\label{eq:proj}
\begin{aligned}
	&\mathrm{MEOP}(A;B):=\Bigl\{\vec{E}=\{\ketbra{\psi_k}{\psi_k}\}_{k=1}^{d^2} \Bigm\vert \braket{\psi_k|\psi_l}=\delta_{kl}, \\
	&\ketbra{\psi_k}{\psi_k} :\mbox{ maximally entangled state on }\cH_A\otimes\cH_B\Bigr\}.
\end{aligned}
\end{align}
Now, we define the followin sets for the construction of PSESs.
\begin{definition}\label{definition:Kr}
	Given a subset $\cP\subset\mathrm{MEOP}(A;B)$ and a parameter $r\ge0$,
	we define the following set of non-positive matrices:
	\begin{align}
		\mathrm{NPM}_r(\cP)
		&:=\Bigl\{\rho=-\lambda E_1+(1+\lambda)E_2+\frac{1}{2}\sum_{k=3}^{d^2}E_k\Big|
		0\le\lambda\le r,\ \vec{E}=\{E_k\}\in\cP
		\Bigr\}.\label{def:NPM}
	\end{align}
	Using  the above set $\mathrm{NPM}_r(\cP)$,
	given a parameter $r\ge0$,
	we define the following two cones
	$\cK^{(0)}_r(\cP)$
	and
	$\cK_r(\cP)$
	as
	\begin{align}
		\cK^{(0)}_r(\cP)
		:&=\mathrm{SES}(A;B)+
		\mathrm{NPM}_r(\cP)
		,\\
		\cK_r(\cP)
		:&=\left(
		\cK^{(0)\ast}_r(\cP)+\mathrm{NPM}_r(\cP)
		\right)^\ast.\label{def:Kr}
	\end{align}
\end{definition}
Then, the following proposition holds.
\begin{proposition}\label{prop:construction1}
	Given $\cH_A$, $\cH_B$,
	define a real number $r_0(A;B)$ as
	\begin{align}\label{def:r0}
		r_0(A;B):=\left(\sqrt{2d}-2\right)/4.
	\end{align}
	When two parameters $r_1$ and $r_2$ satisfy $r_2\le r_1\le r_0(A;B)$,
	two cones
	$\cK_{r_1}(\cP)$
	and
	$\cK_{r_2}(\cP)$
	are pre-dual cones satisfying \eqref{eq:quantum}
	and the inclusion relation
	\begin{align}\label{eq:con-hie}
		\cK_{r_2}(\cP)\subsetneq\cK_{r_1}(\cP).
	\end{align}
\end{proposition}
The proof of proposition~\ref{prop:construction1} is written in appendix~\ref{append-construction}.
Proposition~\ref{prop:construction1} guarantees that
$\cK_r(\cP)$
is pre-dual for any $r\le r_0$.
Therefore, theorem~\ref{theorem:sd} gives a self-dual modification of
$\cK_r(\cP)$
with \eqref{eq:quantum}.
Next, we calculate the value $D(\tilde{\cK}_r(\cP))$.
The following proposition estimates the value $D(\tilde{\cK}_r(\cP))$.
\begin{proposition}\label{prop:construction2}
	Given a parameter $r$ with $0<r\le r_0(A;B)$
	and a self-dual modification
	$\tilde{\cK}_r(\cP)$,
	the following inequality holds:
	\begin{align}\label{eq:est-distance}
		D(\tilde{\cK}_r(\cP))
		\le2\sqrt{\cfrac{2r}{2r+1}}.
	\end{align}
\end{proposition}
The proof of proposition~\ref{prop:construction2} is written in appendix~\ref{append-construction}.
For the latter use, we define the parameter $\epsilon_r$ as
\begin{align}\label{def:er}
	\epsilon_r:=2\sqrt{\cfrac{2r}{2r+1}}.
\end{align}
Proposition~\ref{prop:construction2} implies that the model
$\tilde{\cK_r}(\cP)$
is an $\epsilon_r$-PSES with \eqref{eq:quantum}.
Also,
due to \eqref{eq:con-hie} in proposition~\ref{prop:construction1},
for an arbitrary number $n$,
an exact inequality
\begin{align}\label{eq:ri}
	0<r_n<\cdots<r_1\le r_0(A;B)
\end{align}
gives an exact hierarchy of pre-dual cones
$\{\cK_{r_i}(\cP)\}_{i=1}^n$
with \eqref{eq:quantum}.
Thus,
theorem~\ref{theorem:hie1} gives an independent family
$\{\tilde{\cK}_{r_i}(\cP)\}$
with \eqref{eq:quantum},
and the distance
$D(\tilde{\cK}_{r_i}(\cP))$
is estimated as
\begin{align}
	D(\tilde{\cK}_{r_i}(\cP))
	\le2\sqrt{\cfrac{2r_i}{2r_i+1}}<2\sqrt{\cfrac{2r_1}{2r_1+1}}=\epsilon_{r_1}
\end{align}
by inequalities \eqref{eq:est-distance} and \eqref{eq:ri}.
In other words, the family
$\{\tilde{\cK}_{r_i}(\cP)\}$
is an $n$-independent family of $\epsilon_{r_1}$-PSESs.
Because $n$ is arbitrary and $\epsilon_{r_1}\to0$ holds with $r_1\to0$,
we obtain the following theorem.
\begin{theorem}\label{theorem:main}
	For any $\epsilon>0$,
	there exists an infinite number of $\epsilon$-PSESs.
\end{theorem}
In other words,
there exist infinitely many ESs that cannot be distinguished from the SES by a verification of a maximally entanglement state with small errors
even if the ES is self-dual.

\section{Non-orthogonal discrimination in PSESs}
\label{sect.discrimination}

As the above,
there exist infinite $\epsilon$-PSESs except for the SES
even though $\epsilon$-PSESs are physically similar to the SES.
As the next step,
in this section,
we discuss the operational difference between $\epsilon$-PSESs and the SES in terms of informational tasks.
We focus on the difference between the behaviors of perfect discrimination in
$\tilde{\cK}_r(\cP)$
and $\mathrm{SES}$.
As a result, we show that there exists non-orthogonal perfectly distinguishable states in
$\tilde{\cK}_r(\cP)$ for a certain subset $\cP\subset\mathrm{MEOP}(A;B)$.

In GPTs, perfect distinguishability is defined similarly to quantum theory as follows.
\begin{definition}[perfect distinguishablity]
	Let $\{\rho_k\}_{k=1}^n$ be a family of states $\rho_k\in\cS(\cK,u)$.
	Then, $\{\rho_k\}_{k=1}^n$ are perfectly distinguishable
	if there exists a measurement $\{M_k\}_{k=1}^n\in\cM(\cK,u)$ such that
	$\langle \rho_k,M_l\rangle=\delta_{kl}$.
\end{definition}
The reference \cite{Muller2013} has implied that orthogonality is a sufficient condition for perfectly distinguishablity in any self-dual model under a certain condition in the proof of its main theorem.
Also, the reference \cite{Arai2019} has shown that a model of quantum composite system with non-self-duality has a pair of two distinguishable non-orthogonal  states.
Therefore, it is non-trivial problem whether there exists a self-dual model that has non-orthogonal distinguishable states.
In this section, we show that
any self-dual modification $\tilde{\cK}_r(\cP)$ in section~\ref{sect.construct}
has a measurement to discriminate non-orthogonal states in $\tilde{\cK}_r(\cP)$ perfectly
for a certain subset $\cP\subset\mathrm{MEOP}(A;B)$.

First,
given a vector $\vec{P}=\{P_k\}_{k=1}^{d^2}\in\mathrm{MEOP}(A;B)$,
we define a vector $\vec{E_P}=\{P'_k\}_{k=1}^{d^2}\in\mathrm{MEOP}(A;B)$ as
\begin{align}\label{eq:E-E'}
	P'_1:=P_2,\quad P'_2:=P_1, \quad P'_k=P_k \ (k\ge3).
\end{align}
Then,
given a vector $\vec{P}=\{P_k\}_{k=1}^{d^2}\in\mathrm{MEOP}(A;B)$,
we define a subset $\cP_0(\vec{P})\subset\mathrm{MEOP}(A;B)$ as
\begin{align}\label{def:PE}
	\cP_0(\vec{P}):=\{\vec{P},\vec{E_P}\}.
\end{align}

Now, we consider perfect discrimination in a self-dual modification $\tilde{\cK}_r(\cP_0(\vec{E}))$.
By the equations \eqref{def:NPM} and \eqref{def:PE},
the following two matrices belong to $\mathrm{NPM}_r(\cP_0(\vec{E}))$ for any $\vec{E}$ and any $0\le\lambda\le r$:
\begin{align}
\begin{aligned}
	\label{eq:measurement}
	M_1(\lambda;\vec{P})&:=-\lambda P_1+(1+\lambda)P_2+\frac{1}{2}\sum_{k\ge3}P_k,\\
	M_2(\lambda;\vec{P})&:=-\lambda P_1'+(1+\lambda)P_2'+\frac{1}{2}\sum_{k\ge3}P_k'
	=(1+\lambda) P_1-\lambda P_2+\frac{1}{2}\sum_{k\ge3}P_k,
\end{aligned}
\end{align}
which implies that $M_i(\lambda;\vec{P})\in\cK_r^{(0)\ast}(\cP_0(\vec{P}))\subset\tilde{\cK}_r(\cP_0(\vec{P}))$ for $i=1,2$.
Also, because of the equation \eqref{eq:E-E'},
the equation $M_1(\lambda;\vec{P})+M_2(\lambda;\vec{P})=I$ holds.
Therefore, the family $M(\lambda;\vec{P})=\{M_i(\lambda;\vec{P})\}_{i=1,2}$ is a measurement in $\tilde{\cK}_r(\cP_0(\vec{P}))$ when $0\le\lambda\le r$.

Next, we choose a pair of distinguishable states by $M(\lambda;\vec{P})$.
Let $\ket{\psi_k}$ be a normalized eigenvector of $P_k$.
Then, we define two states $\rho_1,\rho_2$ as follows:
\begin{align}
\begin{aligned}
	\rho_1:&=\ketbra{\phi_1}{\phi_1},\quad \rho_2:=\ketbra{\phi_2}{\phi_2},\\
	\ket{\phi_1}:&=\sqrt{\cfrac{r}{2r+1}}\ket{\psi}_1+\sqrt{\cfrac{r+1}{2r+1}}\ket{\psi}_2,\\
	\ket{\phi_2}:&=\sqrt{\cfrac{r+1}{2r+1}}\ket{\psi}_1+\sqrt{\cfrac{r}{2r+1}}\ket{\psi}_2.
\end{aligned}
\end{align}
Because of the relation $\vec{P}\in\mathrm{MEOP}(A;B)$, the projections $P_i$ and $P_j$ are orthogonal for $i\neq j$,
which implies,
the equations
\begin{align}\label{eq:psi1-2}
	\braket{\psi_i|\psi_j}&=\delta_{i,j},\\
	\braket{\psi_i|P_j|\psi_i}&=\delta_{i,j}.\label{eq:psi-E}
\end{align}
Therefore, the following relation holds for $i,j=1,2$:
\begin{align}
	\Tr \rho_iM_j(r;\vec{P})=\delta_{i,j},
\end{align}
i.e.,
the states $\rho_1$ and $\rho_2$ are distinguishable by the measurement $M(\lambda;\vec{P})$.

Next, we show that $\rho_1,\rho_2\in\tilde{\cK}_r(\cP_0(\vec{P}))$,
which is shown as follows.
Because of the equation
$\mathrm{NPM}_r(\cP_0(\vec{P})):=\{M_i(\lambda;\vec{P}) | 0\le \lambda\le r,\ i=1,2\}$,
any extremal element $x\in\cK_r^{(0)}(\cP_0(\vec{P}))$ can be written as $x=\sigma+M_i(\lambda;\vec{P})$,
where $\sigma\in\mathrm{SES}(A;B)$, $0\le \lambda\le r$, $i=1,2$.
Besides, the following two inequalities hold:
\begin{align}
	\Tr \rho_1 M_1(\lambda;\vec{P})&=-\lambda\cfrac{r}{2r+1}+(1+\lambda)\cfrac{r+1}{2r+1}
	=(\lambda+r+1)\cfrac{1}{2r+1}\stackrel{(a)}{\ge}\cfrac{r+1}{2r+1}\ge0,\label{eq:rho1M1}\\
	\Tr \rho_1 M_2(\lambda;\vec{P})&=-\lambda\cfrac{r+1}{2r+1}+(1+\lambda)\cfrac{r}{2r+1}=(-\lambda+r)\cfrac{1}{2r+1}\stackrel{(b)}{\ge}0.\label{eq:rho1M2}
\end{align}
The equations $(a)$ and $(b)$ are shown by the inequality $0\le \lambda \le r$.
Because the inequality $\Tr \rho_1\sigma\ge0$ holds for any $\sigma\in\mathrm{SES}(A;B)$,
we obtain $\Tr \rho_1 x\ge0$ for any $x\in\cK_r^{(0)}(\cP_0(\vec{P}))$,
which implies $\rho_1\in\cK^{(0)\ast}_r(\cP_0(\vec{P}))$.
Therefore, $\rho_1\in\tilde{\cK}_r(\cP_0(\vec{P}))$ because of the inclusion relation $\cK^{(0)\ast}_r(\cP_0(\vec{P}))\subset \tilde{\cK}_r(\cP_0(\vec{P}))$.
The same discussion derives that $\rho_2\in\tilde{\cK}_r(\cP_0(\vec{P}))$.
As a result,
we obtain a measurement and a distinguishable pair of two states by the measurement in
$\tilde{\cK}_r(\cP_0(\vec{P}))$.

Finally, the following equality implies that $\rho_1$ and $\rho_2$ are non-orthogonal for $r>0$:
\begin{align}
	\Tr\rho_1\rho_2
	&=2\cfrac{r(r+1)}{(2r+1)^2}>0.\label{eq:orthogonal1}
\end{align}
That is to say,
$\rho_1$ and $\rho_2$ are perfectly distinguishable non-orthogonal states.
Here, we apply proposition~\ref{prop:construction2} for the case with
$\epsilon=2\sqrt{(2r)/(2r+1)}$.
Then,
$\tilde{\cK}_r(\cP_0(\vec{P}))$ is an $\epsilon$-PSES that contains a pair of two perfectly distinguishable states  $\rho_1$ and $\rho_2$ with
\begin{align}\label{eq:orthogonal2}
	\Tr\rho_1\rho_2
	&\stackrel{(a)}{\ge}\cfrac{\epsilon^2(\epsilon^2+8)}{32}
\end{align}
if $r$ satisfies $\epsilon=2\sqrt{(2r)/(2r+1)}$.
The inequality $(a)$ is shown by simple calculation as seen in appendix~\ref{append-4-2} (proposition~\ref{prop:dist2}).
We summarize the result as the following theorem.
\begin{theorem}\label{theorem:dist}
	For any $\epsilon>0$,
	there exists an $\epsilon$-PSES that has a measurement and a pair of two perfectly distinguishable states $\rho_1,\rho_2$ with \eqref{eq:orthogonal2}.
\end{theorem}

In this way,
$\epsilon$-PSESs are different from the SES in terms of state discrimination.
This result implies the possibility that orthogonal discrimination can characterize the standard entanglement structure rather than self-duality.
In other words, we propose the following conjecture as a considerable statement, which is a future work.
\begin{conjecture}\label{conj1}
	If a model of the quantum composite system $\cK$ is not equivalent to the SES,
	$\cK$ has a pair of two non-orthogonal states discriminated perfectly by a measurement in $\cK$.
\end{conjecture}

\section{Entanglement structures with group symmetry}\label{sect.symmetry}

As the above,
the SES cannot be distinguished from ESs based on present experimental facts, $\epsilon$-undistinguishability and self-duality.
In contrast to the experimental facts,
we investigate whether there exists other ES with group symmetric conditions than the SES.
As a result,
we clarify that some symmetric conditions characterize the SES uniquely.

In the general setting of GPTs,
symmetric properties of groups play an important role in restricting models to natural one \cite{Muller2013,BLSS2017,Barnum2019}.
In this paper,
for the investigation of an entanglement structure with certain properties,
we introduce the following symmetry (so called \textit{$G$-symmetry}) for a set $X$ under a subgroup $G$ of $\mathrm{GL}(\cV)$:
\begin{itemize}
	\item[(S)] $G$-closed set $X$:  $g(x)\in X$ for any $x \in X$ and any $g\in G$.
\end{itemize}
Also, we say that a set of families $\cX$ is $G$-symmetric
if any element $g\in G$ and any family $\{X_\lambda\}_{\lambda\in\Lambda}\in\cX$ satisfy
$\{g(X_\lambda)\}_{\lambda\in\Lambda}\in\cX$.

With the case $\cV=\cT(\cH_A\otimes\cH_B)$,
a typical subgroup $G$ of $\mathrm{GL}(\cV)$ is the class of global unitary maps given as
\begin{align}
	\mathrm{GU}(A;B):=&\{g\in\mathrm{GL}(\cT(\cH_A\otimes\cH_B)) \mid g(\cdot):=U^\dag (\cdot) U,\nonumber\\
	& U \ \mbox{is a unitary matrix on $\cH_A\otimes\cH_B$}\}.\label{eq:gu}
\end{align}
In terms of physics, the condition $\mathrm{GU}(A;B)$-symmetry means that any unitary map is regarded as a transformation from state to state, i.e., a time-evolution in GPTs.
In other words,
the condition $\mathrm{GU}(A;B)$-symmetry corresponds to the condition of global structure about time-evolutions. 
However,
we remark that 
$\mathrm{GU}(A;B)$-symmetry is an essential condition for the SES,
i.e.,
the assumption of $\mathrm{GU}(A;B)$-symmetry fixes entanglement structures to the standard one.
In fact, the following proposition also holds.
\begin{theorem}\label{prop:global2}
	Assume that a model $\cK$ satisfies \eqref{eq:quantum} and $\mathrm{GU}(A;B)$-symmetric.
	Then, $\cK=\mathrm{SES}(A;B)$.
\end{theorem}
The proof is written in appendix~\ref{append-5-1},
but we mention the essential fact to prove proposition~\ref{prop:global2}.
At first, we say that a group $G$ preserves entanglement structures
if any $g\in G$ and any $x\in\mathrm{SEP}^\ast(A;B)\setminus\mathrm{SEP}(A;B)$ satisfy
$g(x)\in\mathrm{SEP}^\ast(A;B)\setminus\mathrm{SEP}(A;B)$.
As we see in appendix~\ref{append-5-1},
the group $\mathrm{GU}(A;B)$ does not preserve entanglement structures,
and therefore,
$\mathrm{GU}(A;B)$-symmetry is not derived reasonably from local structures.
This is the essential reason why proposition~\ref{prop:global2} holds.
On the other hand,
we remark that 
the local unitary group $\mathrm{LU}(A;B)$ defined as
\begin{align}
\begin{aligned}
	\mathrm{LU}(A;B):=\{g\in\mathrm{GL}(\cT(\cH_A\otimes\cH_B))&\mid g(\cdot):=(U_A^\dag\otimes U_B^\dag) (\cdot) (U_A\otimes U_B)\\
	 U_A,&U_B \ \mbox{are unitary matrices on $\cH_A,\cH_B$}\}\label{eq:lu}
\end{aligned}
\end{align}
preserves entanglement structures.
Similar to $\mathrm{GU}(A;B)$-symmetry,
the condition $\mathrm{LU}(A;B)$-symmetry means that any unitary map on local system is regarded as a time-evolution,
i.e.,
$\mathrm{LU}(A;B)$-symmetry corresponds to a local structure about time-evolutions.

Another important symmetric property is \textit{homogeneity}.
\begin{definition}
	For a positive cone $\cK$ in a vector space $\cV$, define the set $\mathrm{Aut}(\cK)$ as
	\begin{align}
		\mathrm{Aut}(\cK):=\{f\in\mathrm{GL}(\cV) \mid f(\cK)=\cK\}.
	\end{align}
	Then, we say that a positive cone $\cK$ is homogeneous
	if there exists a map $g\in\mathrm{Aut}(\cK)$ for any two elements $x,y\in \cK^\circ$ such that $g(x)=y$.
\end{definition}
A positive cone with self-duality and homogeneity is called a \textit{symmetric cone},
which is essentially classified into finite kinds of cones including the SES \cite{Jordan1934,Koecher1957}.
As shown in the following theorem,
any symmetric cone with \eqref{eq:quantum} is restricted to the SES.
\begin{theorem}\label{prop:sym1}
	Assume a symmetric cone $\cK$ satisfies \eqref{eq:quantum}.
	Then, $\cK=\mathrm{SES}(A;B)$.
\end{theorem}
The proof is written in appendix~\ref{append-4-2}
However, theorem~\ref{prop:sym1} implies that $\mathrm{Aut}(\cK)$ is larger than $\mathrm{GU}(A;B)$ under the condition that $\cK$ is a symmetric cone with \eqref{eq:quantum}.
\begin{proposition}\label{prop:sym2}
	A symmetric cone $\cK$ with \eqref{eq:quantum} satisfies that
	$\mathrm{Aut}(\cK)\supset\mathrm{GU}(A;B)$.
\end{proposition}

Proposition~\ref{prop:sym2} is shown by theorem~\ref{prop:sym1} and the inclusion relation $\mathrm{Aut}(\mathrm{SES}(A;B))\supset\mathrm{GU}(A;B)$.
Since theorem~\ref{prop:sym1} requires symmetry of a larger group than 
theorem~\ref{prop:global2} under the condition \eqref{eq:quantum},
we can conclude that
the assumption of theorem~\ref{prop:sym1} is a mathematically stronger condition than that of theorem~\ref{prop:global2}.

The reference \cite{Muller2013,BLSS2017,Barnum2019} discussed the relation between the symmetry and 
the properties of cones.
Our result is different from their analysis as follows.
The reference \cite{Muller2013,BLSS2017,Barnum2019} assumes \textit{strong symmetry} on a cone,
which is the symmetry based on $\mathrm{Aut}(\cK)$.
As shown in the reference \cite{Muller2013,BLSS2017,Barnum2019},
the strong symmetry with an additional assumption implies 
that the cone is a symmetric cone.
Therefore, the assumption in \cite{Muller2013,BLSS2017,Barnum2019} also implies the inclusion relation $\mathrm{Aut}(\cK)\supset\mathrm{GU}(A;B)$ under the condition \eqref{eq:quantum}.
In this way,
the conditions in the preceding study \cite{Muller2013,BLSS2017,Barnum2019} are stronger than $\mathrm{GU}(A;B)$-symmetry under the condition \eqref{eq:quantum}.

This paper aims to investigate whether there exist other quantum-like structures than the SES,
and we have considered conditions about the behavior of quantum systems.
On the other hand,
for the aim of the derivation of the SES from operational properties and local structures,
$\mathrm{GU}(A;B)$ is not suitable as the above.
Therefore,
it is also an important problem of whether $\mathrm{LU}(A;B)$-symmetry characterize the SES uniquely.
Here, we give the following two important examples:
\begin{enumerate}[(EI)]
	\item $\Gamma(\mathrm{SES}(A;B))$ (where $\Gamma$ is the partial transposition map that transposes Bob's system)
	\item $\cK_r^\ast(\cP)$ (where $\cP$ is an $\mathrm{LU}(A;B)$-symetric subset of $\mathrm{MEOP}(A;B)$)
\end{enumerate}
These two examples satisfy two of three conditions,
$\mathrm{LU}(A;B)$-symmetry, self-duality, and $\epsilon$-undistinguishablity.
That is,
the example (EI) is an $\mathrm{LU}(A;B)$-symmetric self-dual model (proposition~\ref{prop:gamma} in appendix~\ref{append-5-4}),
but (EI) does not satisfy $\epsilon$-undistinguishablity for enough small $\epsilon$.
Also,
the example (EII) satisfies $\mathrm{LU}(A;B)$-symmetry and $\epsilon$-undistinguishablity for any $\epsilon\ge0$, $r=\epsilon^2/(2(4-\epsilon))$, and $\mathrm{LU}(A;B)$-symmetric subset $\cP\subset\mathrm{MEOP}(A;B)$ (because of proposition~\ref{prop:lu} in appendix~\ref{append-5-4}),
but (EII) is not self-dual.
A typical example of $\mathrm{LU}(A;B)$-symmetric subset $\cP$ is $\mathrm{MEOP}(A;B)$.
On the other hand, no known example satisfies the above three conditions except for $\mathrm{SES}(A;B)$.
One might consider that a self-dual modification $\tilde{\cK}_r(\cP)$ satisfies all of the above three condition
because (EII) is pre-dual as seen in section~\ref{sect.construct}.
However, theorem~\ref{theorem:sd} does not ensure that a self-dual modification $\tilde{\cK}$ is $\mathrm{LU}(A;B)$-symmetric even if the pre-dual cone $\cK$ is $\mathrm{LU}(A;B)$-symmetric.
Therefore, it remains open whether there exists a model that satisfies these three conditions and that is different from SES.

\section{Discussion}\label{sect.conclude}
In this paper,
we have discussed whether quantum-like structures exist except for the SES,
even if the structures satisfy experimental facts.
For this aim,
we have considered several conditions for behaving in a similar way as quantum theory.
As the first condition,
we have considered a class of general models, called ESs, whose bipartite local subsystems are equal to quantum systems.
Besides, as an experimental fact,
we have assumed $\epsilon$-undistinguishability
as the success of verification of any maximally entangled state with $\epsilon$ errors.
Also, as another experimental fact,
we have assumed self-duality via projectivity and pre-duality.
Then, we have investigated whether there exists an $\epsilon$-PSES,
i.e.,
self-dual entanglement structure with $\epsilon$-undistinguishability except for the SES.
For this aim,
we have stated general theory,
and we have shown the equivalence between the existence of an independent family of self-dual cones and an exact hierarchy of pre-dual cones.
Then, we have applied the general theories to the problem of whether there exists another model of $\epsilon$-PSESs different from the SES.
Surprisingly, we have shown the existence of an infinite number of independent $\epsilon$-PSESs.
In other words,
we have clarified the existence of infinite possibilities of ESs that cannot
be distinguished from the SES by physical experiments of verification of maximally
entangled states even though the error of verification is extremely tiny and even though
we impose projectivity.
Besides,
we have investigated the difference between the SES and PSESs.
As an operational difference,
we have shown that a certain measurement in our PSESs can discriminate non-orthogonal states perfectly.
Further,
we have explored the possibility of characterizing the SES by a condition about group symmetry.
We have revealed that global unitary symmetry derives the SES from all entanglement structures.
Also, we have shown that any entanglement structure corresponding to a symmetric cone is restricted to the SES.

In this way,
we have shown the existence infinite examples of $\epsilon$-PSESs,
and moreover, our examples are important in the sense that they are self-dual and non-homogeneous.
As we mentioned in the introduction,
the combination of self-duality and homogeneity leads limited types of models including classical and quantum theory.
On the other hand, the reference \cite{Levent2001} presented 
only one example for a self-dual and non-homogeneous model
and a known example is limited to the above one.
Contrary to such preceding studies,
our results have presented infinite examples for self-dual and non-homogeneous models,
which implies a large variety of self-dual models.
This observation also implies that self-duality is not sufficient for the determination of the unique entanglement structure.
For the aim of determination of the unique entanglement structure, we have to impose some additional assumptions to our setting or find a new better assumption,
which denies infinitely many self-dual models of the composite systems.
Our result also has implied the possibility that conjecture~\ref{conj1} resolves the problem,
i.e.,
the combination of self-duality and orthogonal discrimination derives the standard entanglement structure.
Another possibility is occurred from symmetric condition for example $\mathrm{LU}(A;B)$-symmetry,
but it remains open to exist $\mathrm{LU}(A;B)$-symmetric $\epsilon$-PSESs for small $\epsilon$.
These are important open problems.

\section*{Acknowledgement}
HA is
supported by a JSPS Grant-in-Aids for JSPS Research Fellows No. JP22J14947, a JSPS Grant-in-Aids for Scientific Research (B) Grant No. JP20H04139, and a JSPS Grant-in-Aid for JST SPRING No. JPMJSP2125.
MH is supported in part by the National Natural Science Foundation of China (Grant No. 62171212) and
Guangdong Provincial Key Laboratory (Grant No. 2019B121203002).

\newpage

\appendix
\def\thesection{\Alph{section}}

\noindent
{\Large \textbf{Appendix}}

\section{Preparation for the proofs}\label{append-1}

\subsection{Basic properties for cones}

This appendix summarizes fundamental properties about a positive cone that are frequently used throughout this paper.
\begin{proposition}[{\cite{Yoshida2020}[lemma~C.7]}]\label{prop0}
	Let $\cK$ be a positive cone.
	Then, its dual $(\cK^\ast)$ is also a positive cone.
\end{proposition}
\begin{proposition}[{\cite{Plavala2021}[proposition~B.11]}]\label{prop1}
	For any positive cone $\cK$ satisfies
	$(\cK^\ast)^\ast=\cK$.
\end{proposition}
Here, we remark that
proposition~\ref{prop1} enables that
a cone $\cK$ is defined by its dual $\cK^\ast$,
i.e.,
when we want to define a positive cone,
we define its dual cone instead of the positive cone.
The above two propositions are very simple,
but they are not easy to show.
On the other hand,
the following propositions are easy to show by the definition of convex cone.
\begin{proposition}\label{prop2}
	Let $\cK_1\subset\cK_2$ be positive cones.
	Then, $\cK_2^\ast\subset\cK_1^\ast$ holds.
\end{proposition}
\begin{proposition}\label{prop3}
	Let $\cK_1$ and $\cK_2$ be positive cones.
	Then, the dual of the summation of the cones $\left(\cK_1+\cK_2\right)^\ast$
	satisfies
	\begin{align}
		\left(\cK_1+\cK_2\right)^\ast=\cK_1^\ast\cap\cK_2^\ast
	\end{align}
\end{proposition}
\begin{proposition}
	Let $\cK$ be a pre-dual cone.
	Then, any elements $x,y\in\cK^\ast$ satisfy
	$\langle x,y\rangle\ge0$.
\end{proposition}

For the preparation for the following lemma,
we introduce the following condition about a subgroup $G\subset\mathrm{GL}(\cV)$.
\begin{itemize}
	\item[(C)] $G$ is a closed subgroup of $\mathrm{GL}(\cV)$ including the identity map and $G$ satisfies $g^* \in G$ for any element $g \in G$.
\end{itemize}
\begin{lemma}\label{lem:g-closed}
	Let $G$ be a subgroup of $\mathrm{GL}(\cV)$ with (C),
	and let $\cK$ be a positive cone with (S).
	Then, $\cK^\ast$ also satisfies (S).
\end{lemma}
\begin{proof}[Proof of lemma~\ref{lem:g-closed}]
	Take arbitrary elements $x\in\cK^\ast$, $y\in\cK$, and $g\in G$.
	Because $G$ satisfies (C), $g^\ast\in G$.
	Also, because $\cK$ satisfies (S), $g^\ast(y)\in\cK$.
	Therefore, we obtain
	\begin{align}
		\langle g(x),y \rangle=\langle x,g^\ast(y)\rangle\ge 0,
	\end{align}
	i.e.,
	$g(x)\in\cK^\ast$.
\end{proof}
Simply speaking,
lemma~\ref{lem:g-closed} shows that
the condition (C) transmits (S) from $\cK$ to $\cK^\ast$.

\subsection{Definition of limit operations for cones}

For the proof of theorem~\ref{theorem:sd},
we discuss some properties about limit operations for an uncountable family of positive cones.

First, the following proposition guarantees that the intersection of uncountable positive cones is also a positive cone.
\begin{proposition}\label{prop:def-cone1}
	Let $\{\cK_\lambda\}_{\lambda\in\Lambda}$ be a family of positive cones with an uncountably infinite set $\Lambda$.
	There exists a positive cone $\cK$ such that the positive cone $\cK_\lambda$ satisfies $\cK_\lambda\supset\cK$ for any $\lambda\in\Lambda$.
	Then, the set $\bigcap_{\lambda\in\Lambda} \cK_\lambda$ is a positive cone,
	i.e.,
	$\bigcap_{\lambda\in\Lambda} \cK_\lambda$ satisfies the following three conditions:
	\begin{enumerate}[(i)]
		\item $\bigcap_{\lambda\in\Lambda} \cK_\lambda$ is closed and convex.
		\item $\bigcap_{\lambda\in\Lambda} \cK_\lambda$ has an inner point.
		\item $\bigcap_{\lambda\in\Lambda} \cK_\lambda\cap\left(-\bigcap_{\lambda\in\Lambda} \cK_\lambda\right)=\{0\}$.
	\end{enumerate}
\end{proposition}

\begin{proof}
	First, we show (i).
	Because $\cK_\lambda$ is a positive cone for any $\lambda\in\Lambda$,
	$\cK_\lambda$ is closed and convex.
	Because $\cK_\lambda$ is closed, $\bigcap_{\lambda\in\Lambda} \cK_\lambda$ is also closed.
	Take any two elements $x,y\in\left(\bigcap_{\lambda\in\Lambda} \cK_\lambda\right)$.
	Because $x$ and $y$ satisfy $x,y\in\cK_\lambda$ for any $\lambda$,
	$px+(1-p)y\in\cK_\lambda$ for any $p\in[0,1]$,
	which implies that $px+(1-p)y\in\left(\bigcap_{\lambda\in\Lambda} \cK_\lambda\right)$i.e.,
	$\bigcap_{\lambda\in\Lambda} \cK_\lambda$ is convex.
	
	Second, we show (ii).
	Because the assumption $\cK_\lambda\supset\cK$ holds for any $\lambda\in\Lambda$,
	$\left(\bigcap_{\lambda\in\Lambda} \cK_\lambda\right)\supset\cK$ holds.
	Also, because $\cK$ is a positive cone, $\cK$ has an inner point,
	which is also belongs to the interior of $\bigcap_{\lambda\in\Lambda} \cK_\lambda$.
	
	Finally, we show (iii).
	This is shown because the set $\bigcap_{\lambda\in\Lambda} \cK_\lambda\cap\left(-\bigcap_{\lambda\in\Lambda} \cK_\lambda\right)$ is written as
	\begin{align}
		\bigcap_{\lambda\in\Lambda} \cK_\lambda\cap\left(-\bigcap_{\lambda\in\Lambda} \cK_\lambda\right)
		=&\bigcap_{\lambda\in\Lambda} \left(\cK_\lambda\cap(-\cK_\lambda)\right)
		=\bigcap_{\lambda\in\Lambda} \{0\}=\{0\}.
	\end{align}
\end{proof}

Next,
we discuss the sum of uncountable sets.
We remark the definition of  the sum of uncountable sets.

\begin{definition}\label{def:cone-2}
	Let $\{X_\lambda\}_{\lambda\in\Lambda}$ be a family of sets $X_\lambda$ with an uncountably infinite set $\Lambda$.
	We define the set $\sum_{\lambda\in\Lambda} X_\lambda$ as
	\begin{align}
		&\sum_{\lambda\in\Lambda} X_\lambda
		:=\mathrm{Clo}\left(\left\{\sum_{i\in I} x_i\middle| x_i\in X_i, I\subset\Lambda \mbox{ is a finite subset set}\right\}\right),
	\end{align}
	where $\mathrm{Clo}(Y)$ is the closure of a set $Y$.
\end{definition}

Then, the following proposition guarantees that the sum of uncountable positive cones is also a positive cone.

\begin{proposition}\label{prop:def-cone2}
	Let $\{\cK_\lambda\}_{\lambda\in\Lambda}$ be a family of positive cones with an uncountably infinite set $\Lambda$.
	There exists a positive cone $\cK$ such that the positive cone $\cK_\lambda$ satisfies $\cK_\lambda\subset\cK$ for any $\lambda\in\Lambda$.
	Then, the set $\sum_{\lambda\in\Lambda} \cK_\lambda$ is a positive cone,
	i.e.,
	$\sum_{\lambda\in\Lambda} \cK_\lambda$ satisfies the following three conditions:
	\begin{enumerate}[(i)]
		\item $\sum_{\lambda\in\Lambda} \cK_\lambda$ is closed and convex.
		\item $\sum_{\lambda\in\Lambda} \cK_\lambda$ has an inner point.
		\item $\sum_{\lambda\in\Lambda} \cK_\lambda\cap\left(-\sum_{\lambda\in\Lambda} \cK_\lambda\right)=\{0\}$.
	\end{enumerate}
\end{proposition}

\begin{proof}
	First, we show (i).
	By the deinifion~\ref{def:cone-2}, $\sum_{\lambda\in\Lambda} \cK_\lambda$ is closed.
	Take any two elements $x,y\in\left(\sum_{\lambda\in\Lambda} \cK_\lambda\right)$.
	By the definition~\ref{def:cone-2}, the elements $x,y$ are written as $x=\lim_{n\to \infty} x_n$ and $y=\lim_{n\to \infty} y_n$,
	where $x_n,y_n\in\cK_n$.
	Therefore, the element $z_n(p)=px_n+(1-p)y_n$ belongs to $\cK_n$ for $p\in[0,1]$.
	Because $\lim_{n\to\infty} z_n(p)=px+(1-p)y$,
	the element $px+(1-p)y$ belongs to $\sum_{\lambda\in\Lambda} \cK_\lambda$.
	
	Second, we show (ii).
	Because the inclusion relation $\cK_{\lambda_0}\subset\sum_{\lambda\in\Lambda} \cK_\lambda$ holds for any element $\lambda_0\in\Lambda$,
	$\sum_{\lambda\in\Lambda} \cK_\lambda$ has an inner point that is also an inner point of $\cK_{\lambda_0}$.
	
	Finally, we show (iii).
	By the definition~\ref{def:cone-2},
	the element $0$ belongs to $\sum_{\lambda\in\Lambda} \cK_\lambda$.
	Then, we show $\{0\}\supset \sum_{\lambda\in\Lambda} \cK_\lambda\cap\left(-\sum_{\lambda\in\Lambda} \cK_\lambda\right)$.
	Because the assumption $\cK_\lambda\subset\cK$ holds for any $\lambda\in\Lambda$
	and because the positive cone $\cK$ is closed,
	the set $\sum_{\lambda\in\Lambda} \cK_\lambda$ defined as a closure satisfies the inclusion relation
	$\cK\supset\sum_{\lambda\in\Lambda} \cK_\lambda$.
	The inclusion relation $-\cK\supset-\sum_{\lambda\in\Lambda} \cK_\lambda$ also holds.
	Therefore, the following inclusion relation holds:
	\begin{align}
		&\sum_{\lambda\in\Lambda} \cK_\lambda\cap\left(-\sum_{\lambda\in\Lambda} \cK_\lambda\right)
		\subset\cK\cap(-\cK)=\{0\}.
	\end{align}
\end{proof}

Finally,
the following lemma gives a relation between an intersection and a sum over uncountable positive cones.

\begin{lemma}\label{lem:sum-cone}
	Let $\{\cK_\lambda\}_{\lambda\in\Lambda}$ be a family of positive cones with an uncountably infinite set $\Lambda$,
	and let $\cK_1,\cK_2$ be positive cones satisfying $\cK_1\subset\cK_\lambda\subset\cK_2$ for any $\lambda\in\Lambda$.
	Then, the dual cone $\left(\bigcap_{\lambda\in\Lambda} \cK_\lambda\right)^\ast$
	is given by
	$\sum_{\lambda\in\Lambda} \cK_\lambda^\ast$.
\end{lemma}

\begin{proof}[Proof of lemma~\ref{lem:sum-cone}]
	Because of the assumption $\cK_1\subset\cK_\lambda\subset\cK_2$,
	proposition~\ref{prop2} implies  $\cK_1^\ast\supset\cK_\lambda^\ast\supset\cK_2^\ast$.
	Therefore, proposition~\ref{prop:def-cone1} and proposition~\ref{prop:def-cone2} imply that
	the two sets
	$\left(\bigcap_{\lambda\in\Lambda} \cK_\lambda\right)^\ast$ and $\sum_{\lambda\in\Lambda} \cK_\lambda^\ast$ are positive cones.

	Now, we show the duality.
	Take an arbitrary element $y\in\bigcap_{\lambda\in\Lambda} \cK_\lambda$ and
	an arbitrary element  $x\in\sum_{\lambda\in\Lambda} \cK_\lambda^\ast$.
	Then, there exists a sequence $x_n\in\sum_{\lambda\in\Lambda} \cK_\lambda^\ast$ such that
	$\lim_{n\to\infty} x_n=x$.
	Therefore,
	$x,y$ satisfies the following inequality:
	\begin{align}\label{eq:lem:sum-cone}
		\langle x,y\rangle=\lim_{n\to\infty} \langle x_n,y\rangle\ge\lim_{n\to\infty} 0=0,
	\end{align}
	and therefore,
	we obtain the relation $y\in\left(\sum_{\lambda\in\Lambda} \cK_\lambda^\ast\right)^\ast$,
	i.e.,
	$\bigcap_{\lambda\in\Lambda} \cK_\lambda\subset\left(\sum_{\lambda\in\Lambda} \cK_\lambda^\ast\right)^\ast$.
	The opposite inclusion is also shown by the equation \eqref{eq:lem:sum-cone}, similarly.
	Hence, we obtain $\bigcap_{\lambda\in\Lambda} \cK_\lambda=\left(\sum_{\lambda\in\Lambda} \cK_\lambda^\ast\right)^\ast$.
\end{proof}

\section{Another introduction of a model of composite systems}\label{append-com}

Here, we give another meaning of the definition of a model of composite systems.
A model of composite system must satisfy the condition that
any submodel is an exact model of original subsystem when the submodel is extracted from the global model. 
In order to introduce such a concept of a composite system,
we define projections by local effects.
Let $(\cV_A,\cK_A,u_A)$ and $(\cV_B,\cK_B,u_B)$ be two models of GPTs with inner product $\langle~,~ \rangle_A$ and $\langle~,~ \rangle_B$, respectively.
We define the \textit{projection onto $\cV_B$ by the effect $e_A$} as
\begin{align}
\begin{aligned}
	P_{e_A}:&\cV_A\otimes \cV_B\to \cV_B,\\
	P_{e_A}:&x=\sum_k \lambda_k a_k\otimes b_k \mapsto \sum_k\lambda_k\langle a_k, e_A\rangle_A b_k,
\end{aligned}
\end{align}
where $a_k\in\cV_A$ and $b_k\in\cV_B$.
We define the projection onto $\cV_A$ by the effect $e_B$ similarly.

Then, we define a model of composite system in GPTs by the above concept and projections as follows.
We say that a model $(\cV_A\otimes\cV_B,\cK,u_A\otimes u_B)$ with inner product is a composite system of two models $(\cV_A,\cK_A,u_A)$ and $(\cV_B,\cK_B,u_B)$
when $\cK$ satisfies
$P_{e_A}(\cK)=\cK_B$ for any $e_A\in\cK_A^\ast$ and $P_{e_B}(\cK)=\cK_A$ for any $e_B\in\cK_B^\ast$.
This definition derives the inclusion relation (ii) in the definition of the model of the composite system.
\begin{proposition}\label{prop:com}
	Let $\cK$ be a cone on $\cV_A\otimes\cV_B$ satisfying
	$P_{e_A}(\cK)=\cK_B$ for any $e_A\in\cK_A^\ast$ and $P_{e_B}(\cK)=\cK_A$ for any $e_B\in\cK_B^\ast$.
	Then, $\cK$ satisfies $\cK\subset(\cK_A^\ast\otimes\cK_B^\ast)^\ast$,
	where the tensor product of cones is defined as
	\[
		\cK_A\otimes\cK_B:=\left\{\sum_k a_k\otimes b_k\middle| a_k\in\cK_A,\ b_k\in\cK_B\right\}.
	\]
\end{proposition}
\begin{proof}[Proof of proposition~\ref{prop:com}]
	We prove the proposition by contradiction.
	Assume that there exists an element $x\in\cK$ such that $x\not\in(\cK_A^\ast\otimes\cK_B^\ast)^\ast$,
	and $x$ can be written as $\sum_k\lambda_k a_k\otimes b_k$, where $\lambda_k\in\mathbb{R}$, $a_k\in\cK_A$, $b_k\in\cK_B$.
	Then, there exists a separable effect $e\in \cK_A^\ast\otimes\cK_B^\ast$ such that $\langle x,e\rangle<0$.
	Because $\cK_A^\ast\otimes\cK_B^\ast$ is spanned by product elements,
	we choose $e$ as the product element $e_A\otimes e_B$, where $e_A\in\cK^\ast$ and $e_B\in\cK^\ast$,  without loss of generality.
	However, we obtain the following inequality:
	\begin{align}
		0&>\langle x, e_A\otimes e_B\rangle
		=\left\langle \sum_k \lambda_k a_k\otimes b_k e_A\otimes e_B\right\rangle
		=\sum_k\lambda_k \langle a_k,e_A\rangle_A \langle b_k,e_B\rangle_B\nonumber\\
		&=\left\langle \sum_k\lambda_k \langle a_k,e_A\rangle_A b_k,e_B\right\rangle_B
		=\langle P_{e_A}(x),e_B\rangle_B.
	\end{align}
	This inequality implies $P_{e_A}(x)\not\in\cK_B$,
	which contradicts to the assumption.
\end{proof}
In this way,
if the model of the composite system satisfies
the assumption of proposition~\ref{prop:com},
the cone $\cK$ must satisfy $\cK\subset(\cK_A^\ast\otimes\cK_B^\ast)^\ast$.
GPTs also assume such a principle for effect space,
i.e.,
GPTs assume that
the dual also satisfies $\cK^\ast\subset((\cK_A^\ast)^\ast\otimes(\cK_B^\ast)^\ast)^\ast$.
This inclusion is rephrased as
$\cK\supset\cK_A\otimes\cK_B$.
Then, we summarize the definition of a model of the composite system as before.

\section{Proofs of theorem~\ref{theorem:sd} and theorem~\ref{theorem:hie1}}\label{append-2}

\subsection{Proof of theorem~\ref{theorem:sd}}

\begin{proof}[Proof of theorem~\ref{theorem:sd}]
	\textbf{[OUTLINE]}
	First, as STEP1, we define $\cX$ as a set of all pairs of pre-dual cone and its dual cone,
	and we also define an order relation on $\cX$.
	Next, as STEP2,
	we show the existence of a maximal element in $\cX$ by Zorn's lemma,
	i.e.,
	we show that any totally ordered subset $\cB\subset\cX$ has an upper bound in $\cX$.
	Finally, as STEP3, we show that any maximal element corresponds to self-dual cone.\\
	
	\textbf{[STEP1]}
	Definition of $\cX$ and an order relation on $\cX$.
	
	Define the set $\cX$ of all pairs of pre-dual cone and its dual as:
	\begin{align}
	\begin{aligned}\label{eq:set}
		\cX:=\Bigl\{X:=&(K_X,K_X^\ast)\subset\cV\times\cV \Big| \cK\supset K_X, \ K_X \mbox{ is pre-dual cone}\Bigr\}.
	\end{aligned}
	\end{align}
	Also, we define an order relation $\preceq$ on $\cX$ as
	\begin{quote}
		$X\preceq Y \ \Leftrightarrow K_X\supseteq K_Y, \mbox{ and } K_X^\ast \subseteq K_Y^\ast$
		for any $X=(K_X,K_X^\ast)$, $Y=(K_Y,K_Y^\ast)$.
	\end{quote}
	
	\textbf{[STEP2]}
	The existence of the maximal element.
	
The aim of this step is showing
the existence of the maximal element of $\cX$
by applying Zorn's lemma.
For this aim, we need to show the existence of an upper bound
for every totally ordered subset in $\cX$.
That is, it is needed to show 
that the element written as
	\begin{align}
		X':=\left(\bigcap_{B\in\cB} K_B, \Big(\bigcap_{B\in\cB} K_B\Big)^\ast\right)
	\end{align}
is an upper bound in $\cX$ 
for a totally ordered subset $\cB\subset\cX$.
Since any $X\in\cB$ satisfies $X\preceq X'$ by definition of $X'$,
	non-trivial thing is $X'\in\cX$. 
Therefore, we show this membership relation in the following.
	
	Because any $K_B$ satisfies $K_B\supset\cK^\ast$, the subset
	$\bigcap_{B\in\cB}K_B$ has non-empty interior.
	Therefore, it is sufficient to show that $\bigcap_{B\in\cB} K_B$ is pre-dual in order to show $X'\in\cX$.
That is, the condition $X'\in\cX$ follows from 
the relation $\Big(\bigcap_{B\in\cB} K_B\Big)^\ast\subset
\bigcap_{B\in\cB} K_B$.

	For any $X=(\cK_X,\cK_X^\ast)\in\cB$ and $Y=(\cK_Y,\cK_Y^\ast)\in\cB$,
	one of the following inclusion relations holds by totally order of $\cB$:
	\begin{align}
		\cK_X\supseteq \cK_X^\ast \supseteq \cK_Y^\ast  &\quad \left(X\preceq Y\right),\\
		\cK_X\supseteq \cK_Y \supseteq \cK_Y^\ast &\quad \left(Y\preceq X\right).
	\end{align}
	Therefore, $\cK_X\supset\cK_Y^\ast$ holds for any $X,Y\in\cB$,
	which implies $\bigcap_{B\in\cB} K_B\supset K_X$ for any $X\in\cB$.
	Hence, we have
	$\sum_{B\in\cB} K_B^\ast \subset \bigcap_{B\in\cB} K_B$ because
	the set $\bigcap_{B\in\cB} K_B^\ast\supset K_X^\ast$ is a positive cone,
	i.e.,
	closed under linear combination of non-negative scalars.
	Because lemma~\ref{lem:sum-cone} guarantees 
	the relation
	$\Big(\bigcap_{B\in\cB} K_B\Big)^\ast=\sum_{B\in\cB} K_B^\ast$, the above discussion implies
	$X'$ is pre-dual,
	and therefore, we obtain the relation $X'\in\cX$.
	
	Consequently, we have finished showing that every totally ordered in $\cX$ has an upper bound in $\cX$.
	Therefore, Zorn's lemma ensures the existence of the maximal element $X\preceq \tilde{X}\in\cX$.
	\\
	
	\textbf{[STEP3]}
	Self-duality of any maximal element.
	
	We consider maximal element of $X=(\cK,\cK^\ast)$ and write the maximal element as $\tilde{X}=(\tilde{\cK},\tilde{\cK}^\ast)$.
	Here, we will show $\tilde{\cK}$ is self-dual by contradiction.
	Assume $\tilde{\cK}$ is not self-dual, i,e,, $\tilde{\cK}\supsetneq\tilde{\cK}^\ast$,
	and, we take an element $x_0\in\tilde{\cK}\setminus\mathrm{Clo}\left({\cK}^\ast\right)$.
	Then, $\cK'^\ast:=\tilde{\cK}^\ast+\{x_0\}$ satisfies $\cK\supsetneq\cK'$ because $\cK^\ast\subsetneq\cK'^{\ast}$.
	Hence, $\cK'\in\cB$ and $\tilde{X}\prec(\cK',\cK'^\ast)$ hold.
	However, this contradicts to the maximality of $\tilde{X}$.
	As a result, $\tilde{\cK}$ is self-dual.
\end{proof}

\subsection{Proof of theorem~\ref{theorem:hie1}}\label{append-hie}

\begin{proof}
	\textbf{[STEP1]}
	$(i)\Rightarrow(ii)$.
	
	Let $\{\cK_i\}_{i~1}^n$ be an exact hierarchy of pre-dual cones with $\cK\supset\cK_i\supset\cK^\ast$.
	By fixing an element $\rho_i\in\cK_i\setminus\cK_{i+1}$,
	define cones $\cL_i$ as the self-dual modification of $\cK_i'$
	\begin{align}\label{def:cone-sd}
		\cK_i':=\left(\cK_i^\ast + \{\rho_i\}\right)^\ast.
	\end{align}
	Let us show the pre-duality of $\cK_i'$.
	Take any two elements $x',y'\in\cK_i^{\prime\ast}$.
	Because of \eqref{def:cone-sd},
	the elements $x',y'$ is written as $x'=x+\rho_i$, $y'=y+\rho_i$,
	where $x,y\in\cK_i^\ast$.
	Pre-duality of $\cK_i$ implies
	that $\langle x,y \rangle\ge0$.
	Also, the definition of dual implies $\langle x,\rho_i\rangle\ge0$ and $\langle y,\rho_i\rangle\ge0$.
	Because $\langle \rho_i,\rho_i\rangle=||\rho_i||>0$,
	$\langle x',y'\rangle\ge0$ holds,
	which implies that $\cK_i^{\prime\ast}\subset(\cK_i^{\prime\ast})^\ast=\cK_i'$.
	Hence, $\cK_i'$ is a pre-dual cone,
	and Theorem~\ref{theorem:sd} guarantees the existence of a SDM $\tilde{\cK}_i$ satisfying $\cK_i'\supset\tilde{\cK'}_i\supset\cK_i^{\prime\ast}$.
	Also, the definition \eqref{def:cone-sd} implies the inclusion relation $\cK_i^\ast\subset\cK_i^{\prime\ast}$,
	and therefore, we obtain the inclusion relation
	\begin{align}
		\cK_i\supset\cK_i'\supset\tilde{\cK'}_i\supset\cK_i^{\prime\ast}\supset\cK_i^\ast.
	\end{align}
	
	Now, we show the independence of $\{\tilde{\cK'}_i\}$,
	i.e.,
	we show $\tilde{\cK'}_i\not\subset\tilde{\cK'}_j$ and $\tilde{\cK'}_i\not\supset\tilde{\cK'}_j$ for any $i>j$.
	We remark that any two elements $a,b$ in a self-dual cone satisfies $\langle a,b\rangle\ge0$.
	Because $\rho_i$ belongs to $\cK_i\setminus\cK_{i+1}$,
	$\rho_i\not\in\cK_{i+1}\supset\cK_j\supset\tilde{\cK'}_j$ holds,
	which implies $\tilde{\cK'}_i\not\subset\tilde{\cK'}_j$.
	The opposite side $\tilde{\cK'}_i\not\supset\tilde{\cK'}_j$ is shown by contradiction.
	Assume that $\tilde{\cK'}_i\supset\tilde{\cK'}_j$,
	and therefore, Proposition~\ref{prop2} implies $\tilde{\cK'}_i\subset\tilde{\cK'}_j$.
	This contradicts to $\tilde{\cK'}_i\not\subset\tilde{\cK'}_j$.
	As a result, we obtain $\tilde{\cK'}_i\not\supset\tilde{\cK'}_j$.
	\\
	
	\textbf{[STEP2]}
	$(ii)\Rightarrow(i)$.
	
	Let $\{\tilde{\cK}_i\}_{i=1}^n$ be an independent family of self-dual cones with $\cK\supset\tilde{\cK}_i\supset\cK^\ast$.
	Now, we define a cone $\cK_i$ as
	\begin{align}
		\cK_i:=\sum_{j\ge i} \tilde{\cK}_j.
	\end{align}
	The choice of $\cK_i$ implies the inclusion relation $\cK_i\supset\tilde{\cK}_i$,
	i.e.,
	$\tilde{\cK}_i$ is a self-dual modification of $\cK_i$.
	Also, because of the inclusion relation $\cK\supset\tilde{\cK}_i\supset\cK^\ast$,
	the choice of $\cK_i$ implies the inclusion relation $\cK\supset\cK_i\supset\cK^\ast$.
	Moreover, the independence of $\tilde{\cK}_i$ implies the inclusion relation
	\begin{align}
	\cK_i=\sum_{j\ge i} \tilde{\cK}_j\supsetneq\sum_{j\ge i+1} \tilde{\cK}_j=\cK_{i+1},
	\end{align}
	which implies that $\{\cK_i\}_{i=1}$ is an exact hierarchy of pre-dual cones.
\end{proof}

\section{Proofs in section~\ref{sect.construct}}

\subsection{Proof of proposition~\ref{prop:construction1}}\label{append-construction}

For the proof of proposition~\ref{prop:construction1},
we give the following lemmas.

\begin{lemma}\label{lem:con1}
	For given $\cH_A$ and $\cH_B$,
	the relation
	\begin{align}\label{eq:rn}
		\mathrm{NPM}_r(\cP)
		\subset\mathrm{SEP}(A;B)^\ast
	\end{align}
	holds for $0\le r \le (\sqrt{d}-1)/2$.
\end{lemma}

\begin{proof}[Proof of lemma~\ref{lem:con1}]

	The aim of this proof is showing that any element
	$x\in\mathrm{NPM}_r(\cP)$
	satisfies $x\in\mathrm{SEP}^\ast(A;B)$,
	i.e.,
	the element $x$ satisfies $\Tr xy\ge0$ for any $y\in\mathrm{SEP}(A;B)$.
	Take an arbitrary element
	$x\in\mathrm{NPM}_r(\cP)$.
	Then, the element $x$ is written as
	\begin{align}\label{def:Nr}
		x=N(r;\{E_k\}):=-rE_1+(1+r)E_2+\frac{1}{2}\sum_{k=3}^{d^2} E_k,
	\end{align}
	$\{E_k\}\in\cP$
	and $r>0$.
	Here, we remark that any $E_k$ is a maximaly entangled state.
	Therefore, any separable pure state $y$ satisfies the following inequality:
	\begin{align}
		\Tr xy
		=&\Tr y\left(-rE_1+(1+r)E_2+\frac{1}{2}\sum_{j=3}^{d^2} E_j\right)\nonumber\\
		=&\Tr y\left(-\left(r+\frac{1}{2}\right)E_1+\left(r+\frac{1}{2}\right)E_2+\frac{1}{2}I\right)
		\stackrel{(a)}{\ge}\Tr y\left(-\left(r+\frac{1}{2}\right)E_1+\frac{1}{2}I\right)\nonumber\\
		\stackrel{(b)}{\ge}&-\cfrac{r}{\sqrt{d}}-\cfrac{1}{2\sqrt{d}}+\frac{1}{2}
		\stackrel{(c)}{\ge}-\cfrac{\sqrt{d}-1}{2\sqrt{d}}-\cfrac{1}{2\sqrt{d}}+\cfrac{1}{2}=0.
	\end{align}
	The inequality $(a)$ is shown by the inequalities $\Tr y E_1\le1$ and $\Tr y E_2\ge0$.
	The inequality $(b)$ is shown by the fact that the inequality $\Tr\sigma\rho\le(1/\sqrt{d})$ holds for any separable pure state $\sigma$ and any maximally entangled state $\rho$ \cite[Eq. (8.7)]{HayashiBook2017}.
	The inequality $(c)$ is shown by the assumption $0\le r \le (\sqrt{d}-1)/2$.
	Therefore, $\Tr xy\ge0$ holds,
	which implies that $x\in\mathrm{SEP}(A;B)^\ast$.
\end{proof}

\begin{lemma}\label{lem:con2}
	For given $\cH_A$ and $\cH_B$,
	define the dimension $d=\dim(\cH_A)\dim(\cH_B)$.
	Then, any two elements
	$x,y\in\mathrm{NPM}_r(\cP)$
	satisfy $\Tr xy\ge0$ if the parameter $r$ satisfies
	\begin{align}\label{ineq:d}
		0\le r\le \cfrac{\sqrt{2d}-2}{4}.
	\end{align}
\end{lemma}

\begin{proof}[Proof of lemma~\ref{lem:con2}]
	Take two arbitrary elements
	$x,y\in\mathrm{NPM}_r(\cP)$.
	By the definition \eqref{def:NPM},
	the two elements $x,y$ are written as
\begin{align}
\begin{aligned}
	x&=-rE_1+(1+r)E_2+\frac{1}{2}\sum_{k=3}^{d^2} E_k,\\
	y&=-rE_1'+(1+r)E_2'+\frac{1}{2}\sum_{l=3}^{d^2} E_l',
\end{aligned}
\end{align}
where
$\{E_k\},\{E_l'\}\in\cP$.
Then, the following inequality holds:
\begin{align}
	&\Tr xy \nonumber\\
	=&\Tr\left(-rE_1+(1+r)E_2+\frac{1}{2}\sum_{k=3}^{d^2} E_k\right)\left(-rE_1'+(1+r)E_2'+\frac{1}{2}\sum_{l=3}^{d^2} E_l'\right)\nonumber\\
	=&\Tr\left(-rE_1+(1+r)E_2+\frac{1}{2}(I-E_1-E_2)\right)\left(-rE_1'+(1+r)E_2'+\frac{1}{2}(I-E_1'-E_2')\right)\nonumber\\
	=&\Tr\left(-\left(r+\frac{1}{2}\right)E_1+\left(r+\frac{1}{2}\right)E_2+\frac{1}{2}I\right)
	\left(-\left(r+\frac{1}{2}\right)E_1'+\left(r+\frac{1}{2}\right)E_2'+\frac{1}{2}I\right)\nonumber\\
	=&\Tr\Biggl( \left(r+\frac{1}{2}\right)^2E_1E_1'-\left(r+\frac{1}{2}\right)^2E_1E_2'-\frac{1}{2}\left(r+\frac{1}{2}\right)E_1-\left(r+\frac{1}{2}\right)^2E_2E_1'\nonumber\\
	&\quad+\left(r+\frac{1}{2}\right)^2E_2E_2'+\frac{1}{2}\left(r+\frac{1}{2}\right)E_2-\frac{1}{2}\left(r+\frac{1}{2}\right)E_1'+\frac{1}{2}\left(r+\frac{1}{2}\right)E_2'+\frac{1}{4}d \Biggr)\nonumber\\
	\stackrel{(a)}{\ge}&-\left(r+\frac{1}{2}\right)^2-\frac{1}{2}\left(r+\frac{1}{2}\right)-\left(r+\frac{1}{2}\right)^2+\left(r+\frac{1}{2}\right)^2\nonumber\\
	&\quad\quad\quad\quad\quad\quad\quad\quad\quad\quad\quad\quad+\frac{1}{2}\left(r+\frac{1}{2}\right)-\frac{1}{2}\left(r+\frac{1}{2}\right)+\frac{1}{2}\left(r+\frac{1}{2}\right)+\frac{1}{4}d\nonumber\\
	=&-2\left(r+\frac{1}{2}\right)^2+\frac{1}{4}d
	\stackrel{(b)}{\ge}-2\left(\cfrac{\sqrt{2d}-2}{4}+\frac{1}{2}\right)^2+\frac{1}{4}d\nonumber\\
	=&-2\left(\cfrac{\sqrt{2d}}{4}\right)^2+\frac{1}{4}d=-\frac{4d}{16}+\frac{1}{4}d=0.\label{eq:lem:con2}
\end{align}
	The inequality $(a)$ is shown by $E_1E_1'\ge0$, $E_1E_2'\le1$ and so on.
	The inequality $(b)$ is shown by the assumption \eqref{ineq:d} of lemma~\ref{lem:con2}.
	The inequality \eqref{eq:lem:con2} is the desired inequality.
\end{proof}

\begin{proof}[Proof of proposition~\ref{prop:construction1}]
	We remark the following inequality:
	\begin{align}
		\cfrac{\sqrt{2d}-2}{4n}
		&\le \cfrac{\sqrt{2d}-2}{4}
		\le \cfrac{2\sqrt{d}-2}{4}=\cfrac{\sqrt{d}-1}{2}.
	\end{align}
	Therefore, we apply lemma~\ref{lem:con1} and lemma~\ref{lem:con2} for
	$\cK_r(\cP)$
	with $r\le r_0(A;B)$.
	
	First, we show pre-duality of $\cK_r(\cP)$ for $r\le r_0(A;B)$,
	i.e., any two elements $x,y\in\cK_r^\ast(\cP)$ satisfy $\Tr xy\ge0$.
	Take two elements $x,y\in\cK_r^\ast(\cP)$,
	and we need to show $\Tr xy\ge0$.
	Because of the definition~\ref{def:Kr},
	the elements $x,y$ are written as $x=x_1+x_2$ ,$y=y_1+y_2$ for $x_1,y_1\in\cK_r^{(0)\ast}(\cP)$, $x_2,y_2\in\mathrm{NPM}_r(\cP)$.
	By lemma~\ref{lem:con2},
	the inequality $\Tr x_2y_2\ge0$ holds.
	Because $\mathrm{SES}(A;B)\supset\cK_r^{(0)}(\cP)$ holds,
	$\cK_r^{(0)\ast}(\cP)$ is pre-dual,
	and therefore, the inequality $\Tr x_1y_1\ge0$ holds.
	Because 
	$\cK_r^{(0)\ast}(\cP)\subset\mathrm{NPM}_r(\cP)$ holds,
	the inequalities $\Tr x_1y_2\ge0$ and $\Tr y_1x_2\ge0$ hold.
	As a result,
	we obtain $\Tr xy\ge0$, which implies that $\cK_r^\ast(\cP)$ is pre-dual.
	
	Next, we show the exact inclusion relation $\cK_{r_2}(\cP)\subsetneq\cK_{r_1}(\cP)$,
	which is shown by $\mathrm{NPM}_{r_2}(\cP)\subsetneq\mathrm{NPM}_{r_1}(\cP)$ holds.
	Finally, $\cK_r(\cP)$ satisfies \eqref{eq:quantum} because of the definition \eqref{def:Kr} and lemma~\ref{lem:con1}.
\end{proof}

\subsection{Proof of proposition~\ref{prop:construction2}}\label{append-max-ent}

For the proof of proposition~\ref{prop:construction2},
we define a function $F_{\mathrm{max}}$ by fidelity $F(\rho,\sigma)$ of two states $\rho,\sigma$ as
	\begin{align}
		F_{\mathrm{max}}(\rho):&=\max_{\sigma\in\mathrm{ME}(A;B)}F(\rho,\sigma)\stackrel{(a)}{=}\max_{\sigma\in\mathrm{ME}(A;B)}\Tr \rho\sigma.\label{eq:fid-tra}
	\end{align}
The equality $(a)$ holds because any maximally entangled state is pure.
Also, we remark the relation between trace norm and fidelity.
The following inequality holds for any state $\rho,\sigma\in\mathrm{SES}(A;B)$:
\begin{align}\label{eq:F1}
	\|\rho-\sigma\|_1\le 2\sqrt{1-F(\rho,\sigma)}.
\end{align}
In order to show proposition~\ref{prop:construction2},
we give the following lemma.
\begin{lemma}\label{lem:max-ent}
When a state $\rho\in\mathrm{SES}$ and a parameter $r$ satisfy the inequality
	\begin{align}\label{eq:index2}
		F_{\mathrm{max}}(\rho)\le\cfrac{1}{2r+1},
	\end{align}
we have 
	\begin{align}
		\rho\in\cK_r^{(0)\ast}(\cP).
	\end{align}
	\end{lemma}

\begin{proof}[Proof of lemma~\ref{lem:max-ent}]
We choose a state $\rho\in\mathrm{SES}$ and a parameter $r$ to satisfy the inequality \eqref{eq:index2}.
In the following, we show $\rho\in\cK_r^{(0)\ast}(\cP)$.

Any element of $\cK_r^{(0)}(\cP)=\mathrm{SES}(A;B)+\mathrm{NPM}_r(\cP)$ is written as 
$\sigma+N(\lambda;\{E_k\})$ with 
$\sigma \in \mathrm{SES}(A;B)$ and $N(\lambda;\{E_k\})\in \mathrm{NPM}_r(\cP)$ given in \eqref{def:Nr}.
As $\rho \in \mathrm{SES}(A;B)$, we have
\begin{align}
\Tr \rho \sigma \ge 0 \label{HH1}.
\end{align}

	Since the element $N(\lambda;\{E_k\})\in\mathrm{NPM}_r(\cP)$ is written as the following form by $\{E_k\}\in(\cP)$ and $0\le \lambda\le r$
	\begin{align*}
		N(\lambda;\{E_k\})=-\lambda E_1+(1+\lambda)E_2+\frac{1}{2}\sum_{k=3}^{d^2} E_k,
	\end{align*}
	we obtain the following inequality
	by using \eqref{eq:index2}; 
	\begin{align}
		&\Tr\rho N(\lambda;\{E_k\})\nonumber\\
		=&\Tr\rho\left(-\lambda E_1+(1+\lambda)E_2+\frac{1}{2}(I-E_1-E_2)\right)\nonumber\\
		=&\Tr\rho\left(-\left(\lambda+\frac{1}{2}\right) E_1+\left(\lambda+\frac{1}{2}\right)E_2+\frac{1}{2}I\right)
		\stackrel{(a)}{\ge}\Tr\left(-\left(\lambda+\frac{1}{2}\right)\rho E_1+\frac{1}{2}\rho I\right)\nonumber\\
		\stackrel{(b)}{\ge}&-\left(\lambda+\frac{1}{2}\right)\cfrac{1}{2r+1}+\frac{1}{2}
		\stackrel{(c)}{\ge}-\left(r+\frac{1}{2}\right)\cfrac{1}{2r+1}+\frac{1}{2}=0.\label{HH2}
	\end{align}
	The inequality $(a)$ is shown by $\Tr\rho E_2\ge0$.
	The inequality $(b)$ holds because $E_1$ is a maximally entangled state and because the equations \eqref{eq:fid-tra}, \eqref{eq:index2} hold.
	The inequality $(c)$ is shown by $\lambda\le r$.
Therefore, combining \eqref{HH1} and \eqref{HH2},
we obtain
\begin{align}
\Tr \rho (\sigma+N(\lambda;\{E_k\})) \ge 0,
\end{align}
which implies
the relation $\rho\in\cK_r^{(0)\ast}(\cP)$.
\end{proof}

By using lemma~\ref{lem:max-ent},
we prove proposition~\ref{prop:construction2}.

\begin{proof}[Proof of proposition~\ref{prop:construction2}]
	\textbf{[OUTLINE]}
	First, as STEP1, we simplify the minimization of $D(\tilde{\cK_r}(A;B))$.
	Next, as STEP2, we estimate the simplified minimization.
	Finally, as STEP3,
	combining STEP1 and STEP2,
	we derive \eqref{eq:est-distance}.\\
	
	\textbf{[STEP1]}
	Simplification of the minimization.
	
	Because the inclusion relations
	\begin{align}
		\tilde{\cK_r}(\cP)\supset\cK_r^\ast(\cP)\supset\cK_r^{(0)\ast}(\cP)
	\end{align}
	hold,
	the following inequality holds for any $\sigma\in\mathrm{ME}(A;B)$:
	\begin{align}
		D(\tilde{\cK_r}(\cP) \| \sigma)
		=&\min_{\rho\in\tilde{\cK_r}(\cP)}\|\rho-\sigma\|_1
		\le \min_{\rho\in\cK_r^{(0)\ast}(\cP)}\|\rho-\sigma\|_1.\label{eq:prop:const2-1}
	\end{align}
	
	\textbf{[STEP2]}
	Estimation of the minimization.
	
	Given an arbitrary maximally entangled state $\sigma$,
	take an element $\rho_0\in\mathrm{SES}(A;B)$ satisfying the following equality:
	\begin{align}
		F(\rho_0,\sigma)=\cfrac{1}{2r+1}.
	\end{align}
	lemma~\ref{lem:max-ent} implies the relation $\rho_0\in\cK_r^{(0)\ast}(\cP)$.
	Then, we obtain the following inequality:
	\begin{align}
		&\min_{\rho\in\cK_r^{(0)\ast}(\cP)}\|\rho-\sigma\|_1
		\le\|\rho_0-\sigma\|_1
		\stackrel{(a)}{\le} 2\sqrt{1-F(\rho_0,\sigma)}
		=2\sqrt{1-\cfrac{1}{2r+1}}
		=2\sqrt{\cfrac{2r}{2r+1}}.\label{eq:prop:const2-2}
	\end{align}
	The inequality $(a)$ is shown by the inequality \eqref{eq:F1}.
	\\
	
	\textbf{[STEP3]}
	Combination of STEP1 and STEP2.
	
	Because $\sigma$ is an arbitrary element in $\mathrm{ME}(A;B)$,
	the following inequality holds:
	\begin{align}
		&D(\tilde{\cK_r}(\cP))
		=\max_{\sigma\in\mathrm{ME}(A;B)}D(\tilde{\cK_r}(\cP) \| \sigma)
		\stackrel{(a)}{\le} \max_{\sigma\in\mathrm{ME}(A;B)}\min_{\rho\in\cK_r^{(0)\ast}(\cP)}\|\rho-\sigma\|_1
		\stackrel{(b)}{\le}2\sqrt{\cfrac{2r}{2r+1}}\label{eq:d-e}.
	\end{align}
	The inequality $(a)$ is shown by \eqref{eq:prop:const2-1}.
	The inequality $(b)$ holds because the inequality \eqref{eq:prop:const2-2} holds for any $\sigma\in\mathrm{ME}(A;B)$.
	Hence, we obtain \eqref{eq:est-distance}.
\end{proof}

\section{Proof of inequality \eqref{eq:orthogonal2}}\label{append-4-2}

Here,
we show the inequality \eqref{eq:orthogonal2}.
In other words,
we show the following proposition.
\begin{proposition}\label{prop:dist2}
	Let $\epsilon>0$ and $0<r\le r_0(A;B)$ be parameters satisfying $\epsilon=2\sqrt{(2r)/(2r+1)}$,
	and let $\rho_1,\rho_2$ be states satisfying \eqref{eq:orthogonal1}.
	Then, the following inequality holds;
	\begin{align}
		\Tr\rho_1\rho_2\ge\cfrac{\epsilon^2(\epsilon^2+8)}{32}.
	\end{align}
\end{proposition}

\begin{proof}[Proof of proposition~\ref{prop:dist2}]
	The equation $\epsilon=2\sqrt{(2r)/(2r+1)}$ is reduced as follows:
	\begin{align}
		\epsilon=&2\sqrt{(2r)/(2r+1)}\nonumber\\
		\epsilon^2/4=&2r/(2r+1)\nonumber\\
		\epsilon^2/4=&1-1/(2r+1)\nonumber\\
		1/(2r+1)=&(4-\epsilon^2)/4\nonumber\\
		r=&\epsilon^2/(2(4-\epsilon^2)).\label{eq:r-e}
	\end{align}
	Then, \eqref{eq:r-e} implies the following equation:
	\begin{align}
		\Tr \rho_1\rho_2 \stackrel{(a)}{\ge}& \cfrac{2r(r+1)}{(2r+1)^2}
		=\cfrac{2\epsilon^2}{2(4-\epsilon^2)}\cdot\cfrac{\epsilon^2+(2(4-\epsilon^2)}{2(4-\epsilon^2)}\cdot\left(\cfrac{4-\epsilon^2}{4}\right)^2
		=\cfrac{\epsilon^2(\epsilon^2+8)}{32}.
	\end{align}
	Here, we remark that the inequality $(a)$ is \eqref{eq:orthogonal1}.
\end{proof}

\section{Proofs in section~\ref{sect.symmetry}}\label{apend-5}

\subsection{Proof of proposition~\ref{prop:global2}}\label{append-5-1}

\begin{proof}[Proof of proposition~\ref{prop:global2}]
	We show the statement by contradiction.
	Assume that $\cK\neq\mathrm{SES}(A;B)$.
	If $\cK$ satisfies $\cK\subsetneq\mathrm{SES}(A;B)$,
	$\cK$ is not self-dual because of the inclusion relation $\cK\subsetneq\mathrm{SES}(A;B)\subsetneq\cK^\ast$.
	Therefore, we assume the existence of the element $x\in\cK\setminus\mathrm{SES}(A;B)$ without loss of generality.
	Because $x\in\mathrm{SEP}^\ast(A;B)\setminus\mathrm{SES}(A;B)$,
	there exists a pure state $\rho\in\mathrm{SES}(A;B)$ such that $\Tr \rho x<0$.
	Because $\rho$ is pure,
	there exists a unitary map $g\in\mathrm{GU}(A;B)$ such that $g(\rho)\in\mathrm{SEP}(A;B)$.
	Also, because $\cK$ is $\mathrm{GU}(A;B)$-symmetry,
	$g(x)\in\cK$.
	However, $\Tr g(\rho)g(x)=\Tr \rho x<0$,
	and therefore,
	$g(x)\not\in\mathrm{SEP}^\ast(A;B)$.
	This contradicts to $g(x)\in\cK\subset\mathrm{SEP}^\ast(A;B)$.
\end{proof}

\subsection{Proof of theorem~\ref{prop:sym1}}\label{append-5-2}

To show theorem~\ref{prop:sym1}, 
we prepare the classification of symmetric cones.
First, we define the direct sum of cones as follows:
\begin{definition}
	If a family of positive cones $\{\cK_i\}_{i=1}^k$ satisfies $\cK_i\cap\cK_j=\{0\}$ for any $i\neq j$,
	the direct sum of $\cK_i$ is defined as
	\begin{align}
		\bigoplus_{i=1}^k \cK_i:=\left\{\sum_{i=1}^k x_i\middle| x_i\in\cK_i\right\}.
	\end{align}
\end{definition}
Here, we say that a positive cone $\cK$ is irreducible if
the cone $\cK$ cannot be decomposed by a direct sum over more than 1 positive cones as \begin{align}
	\cK=\bigoplus_{i=1}^k \cK_i.
\end{align}
It is known that irreducible symmetric cones are classified into the following five cases \cite{Jordan1934,Koecher1957}:
\begin{inparaenum}[(i).]
	\item $\mathrm{PSD}(m,\mathbb{R})$,
	\item $\mathrm{PSD}(m,\mathbb{C})$,
	\item $\mathrm{PSD}(m,\mathbb{H})$,
	\item $\mathrm{Lorentz}(1,n-1)$,
	\item $\mathrm{PSD}(3,\mathbb{O})$,
\end{inparaenum}
where $n$ and $m$ are arbitrary positive integers, $\mathrm{PSD}(m,\mathbb{K})$ denotes the set of positive semi-definite matrices on a $m$-dimensional Hilbert space over a field $\mathbb{K}$ and $\mathrm{Lorentz}(1,n-1)$ is defined as
\begin{align}
	&\mathrm{Lorentz}(1,n-1)
	:=\{(z,x)\in\mathbb{R}\oplus\mathbb{R}^{n-1}\mid |z|^2\ge|x|^2,\ z\ge0\}.
\end{align}
Besides, it is known that
any symmetric cone $\cK$ can be decomposed by a direct sum over irreducible symmetric cones $\cK_i$ as
\begin{align}
	\cK=\bigoplus_{i=1}^k \cK_i
\end{align} \cite{FKsymmetric}.

Next,
in order to prove theorem~\ref{prop:sym1},
we consider the capacity of a model.
Define the number $\cC(\cK)$ of a model $\cK$ as
the maximum number $m$ of perfectly distinguishable states $\{\rho_k\}_{k=1}^m$ in the model $\cK$.
The following fact is known for the capacity
of entanglement structures.
\begin{fact}[{\cite[proposition4.5]{Yoshida2021}}]\label{prop:cap}
	For any cone $\cK$,
	$\cC(\cK)=\dim(\cH_A\otimes\cH_B)$ holds if $\cK$ satisfies $\mathrm{SEP}(A;B)\subset\cK\subset\mathrm{SEP}^\ast(A;B)$.
\end{fact}

Also, about the symmetric cones in the above list,
preceding studies investigated the capacity of models and the dimension of vector spaces as follows \cite{FKsymmetric}.
In this table,
the dimension is defined as the dimension of vector space including the symmetric cone,
and dim/cap means the ratio of the dimension / the capacity.

\begin{table}[htb]
	\caption{List about irreducible symmetric cones}
	\label{table2}
	\centering
	\begin{tabular}{cccc}
	\hline
	symmetric cone  & capacity  & dimension & dim/cap \\ \hline \hline
	$\mathrm{PSD}(m,\mathbb{R})$ & $m$&$m(m+1)/2$ & $(m+1)/2$ \\ \hline
	$\mathrm{PSD}(m,\mathbb{C})$ &$m$&$m^2$ & $m$  \\ \hline
	$\mathrm{PSD}(m,\mathbb{H})$ &$m$ &$m(2m-1)$ &$2m-1$ \\ \hline
	$\mathrm{Lorentz}(1,n-1)$ &2 &$n$& $n/2$ \\ \hline
	$\mathrm{PSD}(3,\mathbb{O})$ &3 &8 &$8/3$ \\ \hline
	\end{tabular}
\end{table}

Then, we obtain the following lemma.

\begin{lemma}\label{lem:sym}
	Let $\cH_A,\cH_B$ be finite-dimensional Hilbert spaces with dimension larger than 1.
	If an irreducible symmetric cone $\cK$ satisfies \eqref{eq:quantum},
	i.e.,
	$\mathrm{SEP}(A;B)\subset\cK\subset\mathrm{SEP}^\ast(A;B)$,
	the cone $\cK$ is the SES,
	i.e.,
	$\cK=\mathrm{PSD}(m,\mathbb{C})$ for $m=\dim(\cH_A\otimes\cH_B)$.
\end{lemma}

\begin{proof}
	At first, $\cK$ is restricted to the five cases in the list.
	proposition~\ref{prop:cap} implies that $\cK$ has the capacity $m=\dim(\cH_A\otimes\cH_B)$,
	which denies the possibilities $\cK=\mathrm{Lorentz}(1,n-1)$ and $\cK=\mathrm{PSD}(3,\mathbb{O})$
	because $\dim(\cH_A\otimes\cH_B)\ge4$.
	Also, the cone $\cK$ is contained by the vector space of Hermitian matrices on $\cH_A\otimes\cH_B$ with $\mathbb{C}$-valued entries.
	Therefore, the dimension is given by $m^2$.
	Only the case with $\cK = \mathrm{PSD}(m,\mathbb{C})$ satisfies the ratio of the dimension and the capacity of $\cK$
	among the cones listed in Table~\ref{table2}, which shows the desired statement.
\end{proof}

Then, we prove theorem~\ref{prop:sym1}.

\begin{proof}[Proof of theorem~\ref{prop:sym1}]
	First, we decompose $\cK$ by a direct sum over irreducible symmetric cones $\cK_i$ as
	\begin{align}
		\cK=\bigoplus_{i=1}^k \cK_i.
	\end{align}
	Here, we denote the set of extremal points of $X$ by $\mathrm{Ext}(X)$.
	Because each $\mathrm{Ext}(\cK_i)$ is disjoint
	and because any pure state $\rho$ cannot be written as $\rho=\rho_1+\rho_2$ for any Hermitian matrices $\rho_1,\rho_2$ that are not transformed by multiplying any real number,
	the inclusion relation
	\begin{align}
		\mathrm{Ext}(\mathrm{SEP}(A;B))\subset\bigcup_{i=1}^k \mathrm{Ext}(\cK_i)
	\end{align}
	holds.
	Because the set $\mathrm{Ext}(\mathrm{SEP}(A;B))$ is topologically connected,
	$\mathrm{Ext}(\mathrm{SEP}(A;B))$ cannot be written as the disjoint sum of closed sets.
	Therefore, there exists an index $i_0$ such that $\mathrm{Ext}(\mathrm{SEP}(A;B))\subset\mathrm{Ext}(\cK_{i_0})$,
	which implies $\mathrm{SEP}(A;B)\subset\cK_{i_0}$.
	Because $\cK_{i_0}$ is self-dual, the inclusion relation $\mathrm{SEP}^\ast(A;B)\supset\cK_{i_0}$ holds.
	Hence, we apply lemma~\ref{lem:sym} for $\cK_{i_0}$,
	and we obtain $\cK_{i_0}=\mathrm{SES}(A;B)$.
	Thus, we obtain the inclusion relation $\cK\supset\mathrm{SES}(A;B)$,
	which implies $\cK=\mathrm{SES}(A;B)$ because $\cK$ is self-dual.
\end{proof}

\subsection{Relation to the reference \cite{Barnum2019}}\label{append-5-3}

We compare theorem~\ref{prop:sym1}
with the reference \cite{Barnum2019}, which derives properties of a cone from
a certain symmetric condition.
The reference \cite{Barnum2019} introduces the following two assumptions for a positive cone.
\begin{enumerate}
	\item[{[A1]}] (\textit{strong symmetry}) any two tuples of perfectly distinguishable pure states $\{\rho_i\}_{i=1}^n$, $\{\sigma_i\}_{i=1}^n$ are transitive by the map $\mathrm{Aut}(\cK)$,
	i.e.,
	there exists $f\in\mathrm{Aut}(\cK)$ such that $f(\rho_i)=\sigma_i$ holds for any $i$.
	\item[{[A2]}] (\textit{spectrality}) any state $\rho$ has a unique convex decomposition over perfectly distinguishable pure states $\rho_i$ as $\rho=\sum_{i=1}^n p_i\rho_i$ up to permutation of the coefficients $p_i$, 
	where $0\le p_i\le 1$ and $\sum_i p_i=1$.
\end{enumerate}
Then, the reference \cite{Barnum2019} characterizes models with A1 and A2 as follows.
\begin{fact}[\cite{Barnum2019}]\label{theorem:Barnum}
	Assume that a positive cone $\cK$ satisfies A1 and A2.
	Then, $\cK$ is a symmetric cone,
	i.e.,
	$\cK$ is self-dual and homogeneous.
\end{fact}
Therefore, due to proposition~\ref{prop:sym2},
the combination of A1 and A2 implies the inclusion relation $\mathrm{Aut}(\cK)\supset\mathrm{GU}(A;B)$ under the condition \eqref{eq:quantum}.
In other words,
the assumption in theorem~\ref{prop:sym1} is weaker than the assumptions in Fact~\ref{theorem:Barnum} under the condition \eqref{eq:quantum}.

\subsection{Proofs for the examples (EI) and (EII)}\label{append-5-4}

First, we show that the example (EI), i.e., the structure $\Gamma(\mathrm{SES}(A;B))$ is self-dual and $\mathrm{LU}(A;B)$-symmetric.

\begin{proposition}\label{prop:gamma}
	The cone $\Gamma(\mathrm{SES}(A;B))$ is self-dual and $\mathrm{LU}(A;B)$-symmetric.
\end{proposition}

\begin{proof}[Proof of proposition~\ref{prop:gamma}]
	First, we show the self-duality.
	Take arbitrary three elements $x,y\in\Gamma(\mathrm{SES}(A;B))$ and $z\in\mathrm{SEP}^\ast(A;B)\setminus\Gamma(\mathrm{SES}(A;B))$.
	Then, the elements $x,y$ are written as $x=\Gamma(x')$, $y=\Gamma(y')$ for $x',y'\in\mathrm{SES}(A;B)$.
	Because the partial transposition does not change the trace of a product of two matrices, we obtain the following inequality:
	\begin{align}\label{eq:gamma1}
		\Tr xy
		=\Tr \Gamma(x')\Gamma(y')
		=\Tr x'y'\ge0.
	\end{align}
	The last inequality is shown by the self-duality of $\mathrm{SES}(A;B)$.
	Also, the element $z$ is written as $\Gamma(z')$ for $z'\in\mathrm{SEP}^\ast(A;B)\setminus\mathrm{SES}(A;B))$ because any element $w\in\mathrm{SEP}^\ast(A;B)$ satisfies $\Gamma(w)\in\mathrm{SEP}^\ast(A;B)$.
	Then, we obtain the following inequality:
	\begin{align}\label{eq:gamma2}
		\Tr xz
		=\Tr \Gamma(x')\Gamma(z')
		=\Tr x'z'<0.
	\end{align}
	The last inequality is shown by the self-duality of $\mathrm{SES}(A;B)$.
	The two inequalities \eqref{eq:gamma1} and \eqref{eq:gamma2} imply $\Gamma(\mathrm{SES}(A;B))$ is self-dual.
	
	Next, we show $\mathrm{LU}(A;B)$-symmetry.
	Take an arbitrary element $x\in\mathrm{SES}(A;B)$.
	Then, the element $x$ is written as $x=\Gamma(\rho)$ for a state $\rho\in\mathrm{SES}(A;B)$.
	For any unitary matrices $U_A$ and $U_B$,
	the following equation holds:
	\begin{align}
		&\left(U_A^\dag\otimes U_B^\dag\right)\Gamma(\rho)\left(U_A\otimes U_B\right)
		=\Gamma\left(\left(U_A^\dag\otimes U_B^\dag\right)(\rho)\left(U_A\otimes U_B\right)\right).
	\end{align}
	Because $\mathrm{SES}(A;B)$ is $\mathrm{LU}(A;B)$-symmetric,
	the element $\left(U_A^\dag\otimes U_B^\dag\right)(\rho)\left(U_A\otimes U_B\right)$ belongs to $\mathrm{SES}(A;B)$.
	Therefore, 
	the element $\Gamma\left(\left(U_A^\dag\otimes U_B^\dag\right)(\rho)\left(U_A\otimes U_B\right)\right)$ belongs to $\Gamma(\mathrm{SES}(A;B))$,
	which implies that
	$\Gamma(\mathrm{SES}(A;B))$ is $\mathrm{LU}(A;B)$-symmetric.
\end{proof}

Next, we show that the example (EII), i.e., the structure $\cK_r^\ast(\cP)$ is $\mathrm{LU}(A;B)$-symmetric and $\epsilon$-undistinguishable for $r=\epsilon^2/(2(4-\epsilon^2))$ and $\mathrm{LU}(A;B)$-symmetric subset $\cP\subset\mathrm{SES}(A;B)$.

\begin{proposition}\label{prop:lu}
	The cone $\cK_r^\ast(\cP)$ is $\epsilon$-undistinguishable with $r=$ for any subset $\cP\subset\mathrm{MEOP}(A;B)$
	when the parameter $r$ satisfies $r=\epsilon^2/(2(4-\epsilon^2))$.
	Besides,
	the cone $\cK_r^\ast(\cP)$ is $\mathrm{LU}(A;B)$-symmetric when $\cP$ is $\mathrm{LU}(A;B)$-symmetric.
\end{proposition}

\begin{proof}[Proof of proposition~\ref{prop:lu}]

[STEP1]
Derivation of $\epsilon$-undistinguishability.

The inclusion relation $\cK_r^\ast(\cP)\supset\cK_r^{(0)\ast}(\cP)$ implies 

The following inequality implies $\cK_r^\ast(\cP)$ is $\epsilon$-undistinguishable,
i.e., $D(\cK_r^\ast(\cP))\le \epsilon$.
\begin{align}
	D(\cK_r^\ast(\cP))&=\max_{\sigma\in\mathrm{ME}(A;B)}\min_{\rho\in\cK_r^{\ast}(\cP)}\|\rho-\sigma\|_1
	\stackrel{(a)}{\le}\max_{\sigma\in\mathrm{ME}(A;B)}\min_{\rho\in\cK_r^{(0)\ast}(\cP)}\|\rho-\sigma\|_1\nonumber\\
	&\stackrel{(b)}{\le} 2\sqrt{\cfrac{2r}{2r+1}}\stackrel{(c)}{=}\epsilon.
\end{align}
The inequality $(a)$ is shown by the inclusion relation $\cK_r^\ast(\cP)\supset\cK_r^{(0)\ast}(\cP)$.
The inequality is shown by \eqref{eq:d-e}.
The equation $(c)$ is shown by $r=\epsilon^2/(2(4-\epsilon^2))$.
\\

[STEP2]
Derivation of $\mathrm{LU}(A;B)$-symmetry.

First, we show that the set $\mathrm{MNP}_r(\cP)$ is $\mathrm{LU}(A;B)$-symmetric.
Take an arbitrary map $g_1\in\mathrm{LU}(A;B)$ and arbitrary element $x_1\in\mathrm{MNP}_r(\cP)$.
Because of \eqref{def:NPM}, the element $x_1$ is written as $x_1=N(r;\{E_k\})$ given in \eqref{def:Nr}.
Then, the equation $g_1(x_1)=N(r;\{g(E_k)\})$ holds,
and therefore, $g_1(x_1)\in\mathrm{MNP}_r(\cP)$ because the subset $\cP$ is $\mathrm{LU}(A;B)-symmetric$.

Next, we show that $\cK_r^{(0)}(\cP)$ is $\mathrm{LU}(A;B)-symmetric$.
Take an arbitrary map $g_2\in\mathrm{LU}(A;B)$ and arbitrary element $x_2\in\cK_r^{(0)}(\cP)=\mathrm{SES}(A;B)+\mathrm{MNP}_r(\cP)$.
Then, the element $x_2$ can be written as $x_2=\sigma+N(r;\{E_k'\})$,
where $\sigma\in\mathrm{SES}(A;B)$.
Because $\mathrm{SES}(A;B)$ is $\mathrm{LU}(A;B)$-symmetric,
the relation $g_2(\sigma)\in\mathrm{SES}(A;B)$ holds,
and therefore, the relation $g_2(x_2)=g_2(\sigma)+N(r;\{g_2(E_k')\})\in\cK_r^{(0)}(\cP)$ holds.

Finally, we show that $\cK_r^\ast(\cP)$ is $\mathrm{LU}(A;B)-symmetric$.
Because $\mathrm{LU}(A;B)$ satisfies the condition (C) in appendix~\ref{append-1},
and because $\cK_r^{(0)}(\cP)$ is $\mathrm{LU}(A;B)-symmetric$,
we can apply lemma~\ref{lem:g-closed} for $\cK=\cK_r^{(0)\ast}(\cP)$ and $G=\mathrm{LU}(A;B)$,
and we obtain that $\cK_r^{(0)\ast}(\cP)$ is $\mathrm{LU}(A;B)$-symmetric.
Therefore,
the equation $\cK_r^\ast(\cP)=\cK_r^{(0)\ast}(\cP)+\mathrm{MNP}_r(\cP)$
implies that $\cK_r^\ast(\cP)$ is  $\mathrm{LU}(A;B)$-symmetric similarly to the above discussion.
\end{proof}

\end{document}